\documentclass[prx, twocolumn, superscriptaddress]{revtex4-2}

\usepackage{amsmath, amssymb}
\usepackage{braket}
\usepackage{bm}
\usepackage{amsthm}
\usepackage{todonotes}
\usepackage{ulem}

\newcommand{\pdif}[2]{\frac{\partial #1}{\partial #2}}
\newcommand{\pddif}[2]{\frac{\partial^2 #1}{\partial #2^2}}
\DeclareMathOperator{\cnot}{\mathrm{CNOT}}
\DeclareMathOperator{\crz}{\mathrm{CRZ}}
\DeclareMathOperator{\rx}{\mathrm{RX}}

\theoremstyle{plain}
\newtheorem{thm}{Theorem}
\newtheorem*{thm*}{Theorem}
\newtheorem{lem}{Lemma}
\newtheorem*{lem*}{Lemma}

\newtheorem*{cor*}{Corollary}
\theoremstyle{remark}
\newtheorem{rem}{Remark}
\newtheorem*{rem*}{Remark}

\begin{document}

\title[Hamiltonian simulation for hyperbolic partial differential equation]{Hamiltonian simulation for hyperbolic partial differential equations by scalable quantum circuits}

\author{Yuki Sato}
\email[]{yuki-sato@mosk.tytlabs.co.jp}
\affiliation{Toyota Central R\&D Labs., Inc., 1-4-14, Koraku, Bunkyo-ku, Tokyo, 112-0004, Japan}
\affiliation{Quantum Computing Center, Keio University, 3-14-1 Hiyoshi, Kohoku-ku, Yokohama, Kanagawa, 223-8522, Japan}

\author{Ruho Kondo}
\affiliation{Toyota Central R\&D Labs., Inc., 1-4-14, Koraku, Bunkyo-ku, Tokyo, 112-0004, Japan}
\affiliation{Quantum Computing Center, Keio University, 3-14-1 Hiyoshi, Kohoku-ku, Yokohama, Kanagawa, 223-8522, Japan}

\author{Ikko Hamamura}
\thanks{Present Address: NVIDIA G.K., Tokyo 107-0052, Japan.}
\affiliation{IBM Quantum, IBM Research -- Tokyo 19-21 Nihonbashi Hakozaki-cho, Chuo-ku, Tokyo, 103-8510, Japan}

\author{Tamiya Onodera}
\affiliation{IBM Quantum, IBM Research -- Tokyo 19-21 Nihonbashi Hakozaki-cho, Chuo-ku, Tokyo, 103-8510, Japan}
\affiliation{Quantum Computing Center, Keio University, 3-14-1 Hiyoshi, Kohoku-ku, Yokohama, Kanagawa, 223-8522, Japan}

\author{Naoki Yamamoto}
\affiliation{Department of Applied Physics and Physico-Informatics, Keio University, Hiyoshi 3-14-1, Kohoku-ku, Yokohama, Kanagawa, 223-8522, Japan}
\affiliation{Quantum Computing Center, Keio University, 3-14-1 Hiyoshi, Kohoku-ku, Yokohama, Kanagawa, 223-8522, Japan}

\begin{abstract}
Solving partial differential equations for extremely large-scale systems within a feasible computation time serves in accelerating engineering developments. 
Quantum computing algorithms, particularly the Hamiltonian simulations, present a potential and promising approach to achieve this purpose. 
Actually, there are several oracle-based Hamiltonian simulations with potential quantum speedup, but their detailed implementations and accordingly the detailed computational complexities are all unclear. 
This paper presents a method that enables us to explicitly implement the quantum circuit for Hamiltonian simulation; the key technique is the explicit gate construction of differential operators contained in the target partial differential equation discretized by the finite difference method.
Moreover, we show that the space and time complexities of the constructed circuit are exponentially smaller than those of conventional classical algorithms. 
We also provide numerical experiments and an experiment on a real device for the wave equation to demonstrate the validity of our proposed method. 
\end{abstract}

\maketitle

\section{Introduction}

Partial differential equations (PDEs) serve as essential tools for investigating the dynamic behavior of various physical phenomena, including heat conduction, fluid dynamics, and electromagnetic waves~\cite{renardy2004introduction}. 
Solving PDEs for extremely large systems within a reasonable computation time is crucial for accelerating engineering developments in industries. 
Despite remarkable progress in addressing extensive physical systems through the use of supercomputers~\cite{plimpton2019direct, zuo2019simulation}, obtaining solutions within a feasible computation time is still 
intractable.

A potentially promising strategy to substantially reduce the computational expenses for solving PDEs involves the utilization of quantum computing. 
Quantum computing has attracted considerable interest in recent decades as a prospective avenue for achieving dramatically fast computation compared to classical computing.
Although quantum computers currently suffer from limited hardware scalability and less noise resistance, there has been remarkable progress in hardware performance.
One of the promising applications of quantum computers is a solver of PDEs.

For steady-state problems, PDEs reduce to a system of linear or non-linear equations and can be solved by a linear system solver.
There are mainly two types of quantum algorithms for solving systems of linear equations: one is variational quantum algorithms~\cite{cerezo2021variational, tilly2022variational, bravo2023variational} and the other is the Harrow-Hassidim-Lloyd (HHL) algorithm~\cite{harrow2009quantum}.
Variational quantum algorithms are aimed at the use on near-term quantum devices and have been extensively studied for applications to PDEs~\cite{lubasch2020variational, sato2021variational, sato2023variational, ali2023performance}.
On the other hand, HHL-based algorithms focus on the fault-tolerant quantum computers and provide theoretical quantum speedup over classical algorithms under certain conditions~\cite{harrow2009quantum, childs2017quantum, cao2013quantum}.
Although HHL-based algorithms require several oracles to be implemented for state preparation, matrix inversion, and extracting solutions~\cite{aaronson2015read}, there are various studies that can be applied for implementing each part~\cite{babbush2018encoding, mcardle2022quantum, bagherimehrab2023fast}.

For time evolution problems governed by PDEs, there are mainly two types of quantum algorithms as well, i.e., the near-term and long-term algorithms.
As for near-term algorithms, variational quantum simulation~\cite{endo2020variational, wada2022simulating} has been applied to solve PDEs~\cite{demirdjian2022variational, leong2023variational} while Hamiltonian simulation~\cite{nielsen2010quantum} is a counterpart for long-term ones.

Hamiltonian simulation involves implementing a quantum circuit for the time evolution of a quantum system, $\exp(-i \mathcal{H} \tau)$, with the time increment $\tau$ and the Hamiltonian of the target quantum system $\mathcal{H}$; it is also referred to as quantum simulation.
The targeted Hamiltonians are, for example, Ising Hamiltonians~\cite{kim2023evidence, mc2023classically} and molecular Hamiltonians~\cite{loaiza2023reducing}.
Remarkably, implementation of quantum simulation based on the Ising Hamiltonian on a 127-qubit quantum computer has been recently reported, which is a significant contribution demonstrating the utility of quantum computers \cite{kim2023evidence}.
On the other hand, there are also reports of applying Hamiltonian simulation by reducing the governing equations of classical systems to the Schr\"{o}dinger equation~\cite{costa2019quantum, babbush2023exponential, jin2023aquantum, jin2023bquantum, an2023linear, meng2023quantum, miyamoto2024quantum}.
Costa et al.~\cite{costa2019quantum} proposed quantum simulation of the wave equation where the Hamiltonian is given as the incidence matrix of a graph that represents the discretized target space.
Babbush et al.~\cite{babbush2023exponential} generalized this approach and proposed a quantum algorithm for simulating classical coupled oscillators with the rigorous proof of the exponential speedup over any classical algorithms.
Jin et al.~\cite{jin2023aquantum, jin2023bquantum} proposed an interesting approach called \textit{Schr\"{o}dingerization} where the general ordinary differential equation is transformed to the Schr\"{o}dinger equation by the warped phase transformation.
An et al.~\cite{an2023linear} proposed a technique of linear combination of Hamiltonian simulation for simulating general nonunitary dynamics.
Although these results suggest that quantum algorithms may exhibit exponential speedup even for classical system simulations,
these methods rely on an oracle access to the Hamiltonian, which makes their implementation by elementary quantum gates unclear.

This paper proposes a Hamiltonian simulation algorithm for solving a special type of PDEs, i.e., linear hyperbolic PDEs without any source terms, which can be transformed into the Schr\"{o}dinger equation.
Specifically, we derive an explicit quantum circuit representation of time evolution operators for Hamiltonian simulation driven by differential operators. 
Then we apply the proposed method to the advection and wave equations, which are transformed into the Schr\"{o}dinger equation and are discretized by the finite difference method (FDM).
The key technique of implementing the time evolution operators given by differential operators is to diagonalize each term of Hamiltonian using the Bell basis.
This is similar to the idea of the extended Bell measurement~\cite{kondo2022computationally} that efficiently estimates the expectation of band matrices derived from the discretization of PDEs.

The contributions of this paper are summarized as follows. 
\begin{itemize}
    \item We provide an algorithm to build quantum circuits of time evolution operators for Hamiltonian simulation by differential operators.
    This result will contribute to implementing the FDM on a quantum computer.
    
    \item We derive the space and time complexities (Theorems~\ref{thm:complexity} and \ref{thm:complexity_second}); 
    our implementation requires $dn$-qubits and quantum circuits with $O(d n^3 T^2 / \varepsilon)$ or $O(d n^{2.5} T^{1.5} / \varepsilon^{0.5})$ non-local gates, to perform Hamiltonian simulation of hyperbolic PDEs defined on the $d$-dimensional lattice with $2^n$ nodes in each dimension up to time $T$ within the additive error $\varepsilon$. 
    That is, the algorithm enjoys an exponential reduction of resources with respect to the spatial degree of freedom. 
    
    \item We transform the advection and wave equations into the Schr\"{o}dinger equation in the real space for solving them by Hamiltonian simulation, based on the proof of self-adjointness of the Hamiltonian derived from these equations. 
    This analysis is followed by thorough numerical simulations and an experiment on a real quantum device.
\end{itemize}

The rest of this paper is organized as follows.
In Sec.~\ref{sec:preliminary}, we briefly introduce the finite difference operators along with their representation in a qubits system. 
We also discuss the transformation of the advection and wave equations into the Schr\"{o}dinger equation to make our method applicable to these equations.
In Sec.~\ref{sec:method}, we then provide the quantum circuit implementation of the finite difference operators, together with the theoretical error bound between the constructed circuit and the target Hamiltonian evolution. 
In Sec.~\ref{sec:result}, we provide several numerical experiments and an experimental result of a real device.
Finally, we conclude this study in Sec.~\ref{sec:conclusion}.

\section{Preliminaries} \label{sec:preliminary}
\subsection{Finite difference operators and its representation in qubits} \label{sec:diff_mpo}

Let us consider a one-dimensional domain $\Omega:=(0, L)$ where $L$ is the length of the domain.
We discretize the domain $\Omega$ by uniformly distributed $N$ points with the interval $l := L / (N + 1)$, where $N$ is a power of two, i.e., $N = 2^{n}$ for $n \in \mathbb{N}$.
Let us then consider a scalar field $u$ defined on the domain $\Omega$ and discretize the field $u$ using its value at each point (i.e., at each node in the 1-dimensional lattice) as $\bm{u}:=[u_0, u_1, \dots, u_{N-1}]$.
The forward difference operator $D^{+}$ for the spatial derivative acts on $\bm{u}$ as
\begin{align} \label{eq:forward diff operator}
    \left( D^{+} \bm{u} \right)_j = \frac{u_{j+1} - u_j}{l} ~~ \text{ for } j = 0, 1, \dots, N-1,
\end{align}
where $u_{N}$ is determined from the boundary condition (BC). 
For instance,
\begin{align} \label{eq:forward_BC}
    u_{N}:=\begin{cases}
        0 & \text{ for Dirichlet BC}, \\
        u_{N-1} & \text{ for Neumann BC}, \\
        u_0 & \text{ for periodic BC}.
    \end{cases}
\end{align}
Note that the prescribed value for the Neumann boundary condition is set to ensure that the spatial derivative $\left( D^{+} \bm{u} \right)_{N-1}$ on the boundary node is zero.
Similarly, the backward difference operator $D^{-}$ is defined as the operator that acts on $\bm{u}$ as
\begin{align} \label{eq:backward diff operator}
    \left( D^{-} \bm{u} \right)_j = \frac{u_j - u_{j-1}}{l} ~~ \text{ for } j = 0, 1, \dots, N-1,
\end{align}
where $u_{-1}$ is determined from the BC such as 
\begin{align} \label{eq:backward_BC}
    u_{-1}:=\begin{cases}
        0 & \text{ for Dirichlet BC}, \\
        u_0 & \text{ for Neumann BC}, \\
        u_{N-1} & \text{ for periodic BC}.
    \end{cases}
\end{align}
The central difference operator $D^{\pm}$ and the Laplacian operator $D^\Delta$ are the operators acting as
\begin{align}
    \left( D^{\pm} \bm{u} \right)_j &= \frac{u_{j+1} - u_{j-1}}{2l} \text{ for } j = 0, 1, \dots, N-1,  \label{eq:central diff operator} \\
    \left( D^{\Delta} \bm{u} \right)_j &= \frac{u_{j+1} - 2u_j + u_{j-1}}{l^2} \text{ for } j = 0, 1, \dots, N-1, \label{eq:laplacian operator}
\end{align}
where $u_{N}$ and $u_{-1}$ are defined by Eqs.~\eqref{eq:forward_BC} and \eqref{eq:backward_BC}, respectively.

In what follows we quantize the above-described difference operators. 
For this purpose, we need the quantum state corresponding to the discretized field $\bm{u}$, on which those quantized operators act; that is, let $\ket{u}$ be an $n$-qubit state on which $\bm{u}$ is encoded  as 
\begin{align} \label{eq:ket_u}
    \ket{u}:= \sum_{j=0}^{2^n-1} u_j \ket{j},
\end{align}
where $\ket{j}:=\ket{j_{n-1} j_{n-2} \dots j_0}$ with $j_{n-1}, j_{n-2}, \dots, j_0 \in \{0, 1\}$ is the computational basis.
Here, we assume that $\bm{u}$ is normalized, i.e., $\| \bm{u} \|_2 = 1$.
Now, it is known that, for the qubit-based system, the finite difference operators can be represented as matrix product operators (MPOs) using the following three $2 \times 2$ matrices \cite{kiffner2023tensor}:
\begin{equation}
	\sigma_{01} := \begin{bmatrix}
		0 & 1 \\
		0 & 0
	\end{bmatrix}, ~ \sigma_{10} := \begin{bmatrix}
		0 & 0 \\
		1 & 0
	\end{bmatrix}, ~ I := \begin{bmatrix}
		1 & 0 \\
		0 & 1
	\end{bmatrix},
\end{equation}
where $\sigma_{01}$ and $\sigma_{10}$ are the ladder operators.
The following two $2 \times 2$ matrices are also useful:
\begin{align}
    \sigma_{00} := \begin{bmatrix}
        1 & 0 \\
        0 & 0
    \end{bmatrix}, \quad \sigma_{11} := \begin{bmatrix}
        0 & 0 \\
        0 & 1
    \end{bmatrix}.
\end{align}
The point of quantization is the MPO representation of the shift operators $S^{-}=\sum_{j=1}^{2^n-1} \ket{j-1}\bra{j}$ and $S^{+}=\sum_{j=1}^{2^n-1} \ket{j}\bra{j-1}$ as follows: 
\begin{align}
    S^{-} &:= \sum_{j=1}^{n} I^{\otimes (n-j)} \otimes \sigma_{01} \otimes \sigma_{10}^{\otimes (j-1)} \nonumber \\
        &= \sum_{j=1}^{n} s^{-}_j, \label{eq:shift_m} \\
    S^{+} &:= \left( S^{-} \right)^\dagger \nonumber \\
    &= \sum_{j=1}^{n} I^{\otimes (n-j)} \otimes \sigma_{10} \otimes \sigma_{01}^{\otimes (j-1)} \nonumber \\
    &= \sum_{j=1}^{n} s^{+}_j, \label{eq:shift_p}
\end{align}
where
\begin{align}
    s^{-}_j &:= I^{\otimes (n-j)} \otimes \sigma_{01} \otimes \sigma_{10}^{\otimes (j-1)},  \\
    s^{+}_j &:= I^{\otimes (n-j)} \otimes \sigma_{10} \otimes \sigma_{01}^{\otimes (j-1)}.
\end{align}
Here, $\sigma_{ij}^{\otimes 0}$ is regarded as a scalar 1. 
For example, for the case $n=2$, we have $S^{-}=\ket{0}\bra{1}+\ket{1}\bra{2}+\ket{2}\bra{3}
=\sigma_{01}\otimes\sigma_{10}+I\otimes\sigma_{01}$. 
Using the shift operator, we obtain the forward difference operator with the Dirichlet BC as follows: 
\begin{align}
    D^{+}_\mathrm{D} := \frac{1}{l} \left( S^{-} - I^{\otimes n} \right).
\end{align}
Actually, the following relationship holds:
\begin{align}
    D^{+}_\mathrm{D} \ket{u} &= \frac{1}{l} \left( S^{-} - I^{\otimes n} \right) \sum_{j=0}^{2^n-1} u_j \ket{j} \nonumber \\
    &= \frac{1}{l} \left( \sum_{j=1}^{2^n-1} u_j \ket{j-1} - \sum_{j=0}^{2^n-1} u_j \ket{j} \right) \nonumber \\
    &= \sum_{j=0}^{2^n-1} \frac{u_{j+1} - u_j}{l} \ket{j},
\end{align}
where $u_{2^n} = 0$, which exactly corresponds to the forward difference operator with the Dirichlet BC. 
To impose the Neumann BC, it suffices to add $\sigma_{11}^{\otimes n} / l$ to $D^{+}_\mathrm{D}$ as 
\begin{align}
    D^{+}_\mathrm{N} := \frac{1}{l} \left( S^{-} - I^{\otimes n} + \sigma_{11}^{\otimes n} \right).
\end{align}
For the periodic BC, we have to add $\sigma_{10}^{\otimes n} / l$ to $D^{+}_\mathrm{D}$ as
\begin{align}
    D^{+}_\mathrm{P} := \frac{1}{l} \left( S^{-} - I^{\otimes n} + \sigma_{10}^{\otimes n} \right).
\end{align}
Similarly, the other finite difference operators can be represented as
\begin{align}
    &\begin{cases}
    D^{-}_\mathrm{D} &= \frac{1}{l} \left( I^{\otimes n} - S^{+} \right), \\
    D^{-}_\mathrm{N} &= \frac{1}{l} \left( I^{\otimes n} - S^{+} - \sigma_{00}^{\otimes n} \right), \\
    D^{-}_\mathrm{P} &=  \frac{1}{l} \left( I^{\otimes n} - S^{+} - \sigma_{01}^{\otimes n} \right), 
    \end{cases} \\
    &\begin{cases}
    D^{\pm}_\mathrm{D} &= \frac{1}{2l} \left( S^{-} - S^{+} \right), \\
    D^{\pm}_\mathrm{N} &= \frac{1}{2l} \left( S^{-} - S^{+} - \sigma_{00}^{\otimes n} + \sigma_{11}^{\otimes n} \right), \\
    D^{\pm}_\mathrm{P} &= \frac{1}{2l} \left( S^{-} - S^{+} - \sigma_{01}^{\otimes n} + \sigma_{10}^{\otimes n} \right), 
    \end{cases} \label{eq:quantized central diff operator} \\
    &\begin{cases}
    D^{\Delta}_\mathrm{D} &= \frac{1}{l^2} \left( S^{-} + S^{+} - 2I^{\otimes n} \right), \\
    D^{\Delta}_\mathrm{N} &= \frac{1}{l^2} \left( S^{-} + S^{+} - 2I^{\otimes n} + \sigma_{00}^{\otimes n} + \sigma_{11}^{\otimes n} \right), \\
    D^{\Delta}_\mathrm{P} &= \frac{1}{l^2} \left( S^{-} + S^{+} - 2I^{\otimes n} + \sigma_{01}^{\otimes n} + \sigma_{10}^{\otimes n} \right). 
    \end{cases} \label{eq:laplacian}
\end{align}
Since we have the shift operators $S^{-}$ and $S^{+}$, we can also construct the finite difference operators of the higher-order approximations as discussed in Appendix~\ref{sec:higher-order-FDM}.

Quantization of the difference operators for the $d$-dimensional domain $\Omega=(0, L)^d$ is straightforward, as follows.
We discretize each segment $(0, L)$ by uniformly distributed $N$ points with the interval $l := L / (N + 1)$, which results in a $d$-dimensional lattice with $N^d$ nodes.
We again assume that $N$ is a power of two, i.e., $N = 2^n$ for $n \in \mathbb{N}$.
Let $\ket{u(t)}$ be a $dn$-qubit state on which the discretized field $\bm{u}$ on the lattice is encoded as 
\begin{align} \label{eq:u_d_dim}
    \ket{u(t)} := \sum_{j_1=0}^{2^n-1} \dots \sum_{j_d=0}^{2^n-1} u(t, x_{j_1}, \dots, x_{j_d}) \ket{j_1} \otimes \dots \otimes \ket{j_d}.
\end{align}
where $\ket{j_\alpha}:=\ket{(j_\alpha)_{n-1} (j_\alpha)_{n-2} \dots (j_\alpha)_0}$ with $(j_\alpha)_{n-1}, (j_\alpha)_{n-2}, \dots, (j_\alpha)_0 \in \{0, 1\}$ is the computational basis and $x_{j_\alpha}$ is the spatial coordinate of the $j_\alpha$-th node along the $x_\alpha$-axis.
We again assume that $\bm{u}$ is normalized, i.e., $\| \bm{u} \|_2=1$.
The finite difference operator $D^\mu_\mathrm{B}$ defined above can easily be extended to those for $\ket{u} $ in the $d$-dimensional space, as follows:
\begin{align} \label{eq:differential_operator_ddim}
    (D^\mu_\mathrm{B})_\alpha = I^{\otimes (\alpha-1)n} \otimes D^\mu_\mathrm{B} \otimes I^{\otimes (d-\alpha)n},
\end{align}
where $\mu \in \{-, +, \pm, \Delta \}$ and $\mathrm{B} \in \{ \mathrm{D}, \mathrm{N}, \mathrm{P} \}$.

\subsection{Transforming partial differential equations to Schr\"{o}dinger equation} \label{sec:to_schroudinger}

In this study, we particularly focus on the advection equation and the wave equation as hyperbolic partial differential equations.
To apply Hamiltonian simulation for these equations, we first need to express them in the form of the Schr\"{o}dinger equation.

\subsubsection{Advection equation} \label{sec:advect_sch}

Let the scalar field $u$ be governed by the advection equation with the constant velocity field $\bm{v}$ as
\begin{align}
    \pdif{u(t, \bm{x})}{t} + \bm{v} \cdot \nabla u(t, \bm{x}) = 0,
\end{align}
where $t$ is the time, $\bm{x}$ is the spatial coordinate and $\nabla$ is the spatial differential operator; 
that is, $\nabla=\partial/\partial x$ for the one-dimensional case $d=1$. 
The scalar field $u$ represents, for example, temperature or concentration. 
The advection equation can be rewritten as
\begin{align} \label{eq:advect_sch}
    \pdif{u(t, \bm{x})}{t} = -i \left( -i \bm{v} \cdot \nabla \right) u(t, \bm{x}).
\end{align}
If the operator $-i\bm{v} \cdot \nabla$ is self-adjoint (meaning that 
$\left\langle \tilde{u}, \left(-i \bm{v} \cdot \nabla \right) u \right\rangle = \left\langle \left(-i \bm{v} \cdot \nabla \right)\tilde{u},  u \right\rangle$ holds for arbitrary scalar fields $u$ and $\tilde{u}$), Eq.~\eqref{eq:advect_sch} is exactly the Schr\"{o}dinger equation; actually, this property holds under an appropriate BC as shown in Appendix~\ref{sec:self-adjoint_advect}.
Thus, the advection equation falls into the Schr\"{o}dinger wave equation with the Hamiltonian that is expressed in terms of differential operators, $\nabla$.

The next step is to discretize the field variable $u$ and the Hamiltonian $-i \bm{v} \cdot \nabla$. 
That is, $u$ is discretized as $\bm{u}:=[u_0, u_1, \dots, u_{N-1}]$ and each element is encoded into the amplitude of the quantum state $\ket{u(t)}$ given in Eq.~\eqref{eq:u_d_dim}. 
Also, the Hamiltonian $-i \bm{v} \cdot \nabla$ is discretized using the central difference operator \eqref{eq:quantized central diff operator} as
\begin{align} \label{eq:hamiltonian_advect}
    \mathcal{H} = \begin{cases}
        - i v_1 (D^{\pm})_1 & \text{for } d=1, \\
        - i v_1 (D^{\pm})_1 - i v_2 (D^{\pm})_2 & \text{for } d=2, \\
        - i v_1 (D^{\pm})_1 - i v_2 (D^{\pm})_2 - i v_3 (D^{\pm})_3 & \text{for } d=3.
    \end{cases}
\end{align}
Since $(D^{\pm})^\dagger = - D^{\pm}$, the Hamiltonian with the central difference operator is actually a Hermitian matrix. 
As a result, we obtain the Schr\"{o}dinger equation $d \ket{u(t)}/d t = -i \mathcal{H} \ket{u(t)}$, and this is what we aim to simulate on a quantum device using the Hamiltonian simulation method.
Since the initial condition for the advection equation is given as specifying $u(0, \bm{x})$, the state preparation oracle for the Hamiltonian simulation has to prepare a quantum state of $\ket{u(0)}=\sum_{j_1=0}^{2^n-1} \dots \sum_{j_d=0}^{2^n-1} u(0, x_{j_1}, \dots, x_{j_d}) \ket{j_1} \otimes \dots \otimes \ket{j_d}$.

Note that the upwind differencing scheme is preferable for the numerical stability for the advection equation~\cite{allaire2007numerical}, but using $D^{+}$ or $D^{-}$ to discretize $\nabla$ does not retain the Hermitian property of $\mathcal{H}$~\cite{brearley2024quantum}. 
In fact, discretizing $\nabla$ using $D^{+}$ makes the Hamiltonian $\mathcal{H} = -i v_1 D^{+}$ for $d=1$, which is no longer Hermitian because $\mathcal{H}^\dagger = i v_1 (D^{+})^\dagger = -i v_1 D^{-} \neq \mathcal{H}$ with the relationship of $(D^{+})^\dagger = D^{-}$.
This is the reason why we here use the central difference scheme; we will investigate the relationship between finite difference schemes and the resulting Hermitian property to explore the possibility of using $D^{+}$ or $D^{-}$, in our future research.

\subsubsection{Wave equation} \label{sec:wave_sch}
Let $u$ be governed by the wave equation with the speed $c$ as
\begin{align}
    \pddif{u(t, \bm{x})}{t} = c^2 \nabla^2 u(t, \bm{x}).
\end{align}
The scalar field $u$ corresponds to, for example, the displacement of string and membrane, and the pressure.
The wave equation can be rewritten as
\begin{align} \label{eq:wave_sch}
    \pdif{\bm{\psi}(t, \bm{x})}{t} = -i \mathcal{H} \bm{\psi}(t, \bm{x}),
\end{align}
by setting, when $d=1$,
\begin{align}
    \bm{\psi}(t, x_1) = \begin{pmatrix}
        \pdif{u(t, x_1)}{t} \\
	i c\pdif{u(t, x_1)}{x_1}
    \end{pmatrix}, ~ \mathcal{H} = c\begin{pmatrix}
        0 & \pdif{}{x_1} \\
        -\pdif{}{x_1} & 0
    \end{pmatrix},
\end{align}
when $d=2$, 
\begin{align} \label{eq:psi_h_wave_2}
    &\bm{\psi}(t, \bm{x}) = \begin{pmatrix}
        \pdif{u(t, \bm{x})}{t} \\
        i c \left( \pdif{u(t, \bm{x})}{x_1} + i \pdif{u(t, \bm{x})}{x_2} \right)
    \end{pmatrix}, \nonumber \\
    &\mathcal{H} = c\begin{pmatrix}
	0 & \pdif{}{x_1} - i\pdif{}{x_2} \\
	-\pdif{}{x_1} - i\pdif{}{x_2} & 0
    \end{pmatrix},
\end{align}
and when $d=3$,
\begin{align}
    \bm{\psi}(t, \bm{x}) = \begin{pmatrix}
	\pdif{u(t, \bm{x})}{t} \\
	i c\pdif{u(t, \bm{x})}{x_1} \\
        i c\pdif{u(t, \bm{x})}{x_2} \\
        i c\pdif{u(t, \bm{x})}{x_3}
    \end{pmatrix}, ~ \mathcal{H} = c\begin{pmatrix}
    0 & \pdif{}{x_1} & \pdif{}{x_2} & \pdif{}{x_3} \\
    -\pdif{}{x_1} & 0 & 0 & 0 \\
    -\pdif{}{x_2} & 0 & 0 & 0 \\
    -\pdif{}{x_3} & 0 & 0 & 0 \\
    \end{pmatrix},
\end{align}
where $d$ is the number of spatial dimensions and $x_\alpha$ is the spatial coordinate.
This form of the Hamiltonian is similar to that in Ref.~\cite{costa2019quantum}.
Note that when $d=2$, we can also set $\bm{\psi}(t, \bm{x})$ and $\mathcal{H}$ by a $3 \times 1$ vector and a $3 \times 3$ matrix, respectively, which are obtained by omitting the row and column regarding $x_3$-axis of those when $d=3$.
However, Eq.~\eqref{eq:psi_h_wave_2} has the advantage in quantizing these quantities because of the size of power of two.

We can show that, for the above three cases, $\mathcal{H}$ is self-adjoint (meaning that $\langle \tilde{\bm{\psi}}, \mathcal{H} \bm{\psi} \rangle = \langle \mathcal{H} \tilde{\bm{\psi}}, \bm{\psi} \rangle$ 
holds for arbitrary vector fields $\bm{\psi}$ and $\tilde{\bm{\psi}}$) under appropriate BCs; the proof is given in Appendix~\ref{sec:self-adjoint_wave}. 
Therefore, Eq.~\eqref{eq:wave_sch} is exactly the Schr\"{o}dinger wave equation; in other words, the wave equation falls into the Schr\"{o}dinger wave equation with its Hamiltonian containing the differential operators.

The scalar field $u(t, \bm{x})$ is discretized and encoded in a quantum state $\ket{\psi(t)}$ on qubits as
\begin{widetext}
\begin{align}
    \ket{\psi(t)} := \begin{cases}
    \ket{0} \otimes \sum_{j_1=0}^{2^n-1} \pdif{u(t, x_{j_1})}{t} \ket{j_1} + \ket{1} \otimes ic \sum_{j_1=0}^{2^n-1} \pdif{u(t, x_{j_1})}{x_1} \ket{j_1}, & \text{ for } d=1, \\[2mm]
    \ket{0} \otimes \sum_{j_1=0}^{2^n-1} \sum_{j_2=0}^{2^n-1} \pdif{u(t, x_{j_1}, x_{j_2})}{t} \ket{j_1} \otimes \ket{j_2} & \\[2mm]
    \quad + \ket{1} \otimes ic \sum_{j_1=0}^{2^n-1} \sum_{j_2=0}^{2^n-1} \left( \pdif{u(t, x_{j_1}, x_{j_2})}{x_1} + i \pdif{u(t, x_{j_1}, x_{j_2})}{x_2} \right) \ket{j_1} \otimes \ket{j_2}, & \text{ for } d=2, \\[2mm]
    \ket{0} \otimes \ket{0} \otimes \sum_{j_1=0}^{2^n-1} \sum_{j_2=0}^{2^n-1} \sum_{j_3=0}^{2^n-1} \pdif{u(t, x_{j_1}, x_{j_2}, x_{j_3})}{t} \ket{j_1} \otimes \ket{j_2} \otimes \ket{j_3} & \\[2mm]
    \quad + \ket{0} \otimes \ket{1} \otimes ic \sum_{j_1=0}^{2^n-1} \sum_{j_2=0}^{2^n-1} \sum_{j_3=0}^{2^n-1} \pdif{u(t, x_{j_1}, x_{j_2}, x_{j_3})}{x_1} \ket{j_1} \otimes \ket{j_2} \otimes \ket{j_3} & \\[2mm]
    \quad + \ket{1} \otimes \ket{0} \otimes ic \sum_{j_1=0}^{2^n-1} \sum_{j_2=0}^{2^n-1} \sum_{j_3=0}^{2^n-1} \pdif{u(t, x_{j_1}, x_{j_2}, x_{j_3})}{x_2} \ket{j_1} \otimes \ket{j_2} \otimes \ket{j_3} & \\[2mm]
    \quad + \ket{1} \otimes \ket{1} \otimes ic \sum_{j_1=0}^{2^n-1} \sum_{j_2=0}^{2^n-1} \sum_{j_3=0}^{2^n-1} \pdif{u(t, x_{j_1}, x_{j_2}, x_{j_3})}{x_3} \ket{j_1} \otimes \ket{j_2} \otimes \ket{j_3},
& \text{ for } d=3. \\[2mm]
    \end{cases}
\end{align}
\end{widetext}
Note that the initial quantum state $\ket{\psi(0)}$ has to be implemented to satisfy the initial conditions of the wave equation, which are typically the initial conditions of $u(0, \bm{x})$ and $\partial u(0, \bm{x}) / \partial t$.
When we have an initial value of $u(0, x)$ as an analytic function, which is a typical case in solving PDEs, we can calculate $\partial u(0, x)/\partial x$ for the initial condition of our algorithm.
Thus, we assume that we can set the initial condition for the quantum state (i.e., $\partial u(0, x)/\partial x$ and $\partial u(0, x)/\partial t$) based on the typical initial conditions of the wave equation (i.e., $u(0, x)$ and $\partial u(0, x)/\partial t$).
Also, the Hamiltonian $\mathcal{H}$ is discretized by the forward and backward difference operators so that the resulting $\mathcal{H}$ can be actually a Hermitian matrix, as follows:
\begin{widetext}
\begin{align} \label{eq:hamiltonian_wave}
    \mathcal{H} = \begin{cases}
        c \left( \sigma_{01} \otimes (D^{+})_1 - \sigma_{10} \otimes (D^{-})_1 \right), & \text{for } d=1, \\[2mm]
        c \left( \sigma_{01} \otimes \left( (D^{+})_1 - i (D^{+})_2 \right) - \sigma_{10} \otimes \left( (D^{-})_1 + i (D^{-})_2 \right) \right), & \text{for } d=2, \\[2mm]
        c \left( \sigma_{00} \otimes \sigma_{01} \otimes (D^{+})_1 + \sigma_{01} \otimes \sigma_{00} \otimes (D^{+})_2 + \sigma_{01} \otimes \sigma_{01} \otimes (D^{+})_3 \right. & \\[2mm]
        \quad \left. - \sigma_{00} \otimes \sigma_{10} \otimes (D^{-})_1 - \sigma_{10} \otimes \sigma_{00} \otimes (D^{-})_2 - \sigma_{10} \otimes \sigma_{10} \otimes (D^{-})_3 \right), & \text{for } d=3.
    \end{cases}
\end{align}
\end{widetext}
As a result, we obtain the Schr\"{o}dinger equation $d \ket{\psi(t)}/d t = -i \mathcal{H} \ket{\psi(t)}$, which can be simulated on a quantum device. 
Note that the above discretization imposes the Dirichlet BC for $\partial u / \partial t$ at $x_\alpha=0$ and for $\partial u / \partial x_\alpha$ at $x_\alpha=L$. 
That is, the mixed BC for $u$ consisting of the Dirichlet BC for $x_\alpha=0$ and the Neumann BC for $x_\alpha=L$ is imposed, which is consistent with the condition of self-adjointness of the Hamiltonian $\mathcal{H}$ as discussed in Appendix~\ref{sec:self-adjoint_wave}.
We can also use the central difference scheme for periodic BC because it retains the Hermitian property.
We would like to conduct our future research to discuss the scheme for implementing arbitrary BCs while keeping the Hermitian property of the Hamiltonian.

\section{Circuit implementation method for PDEs} \label{sec:method}

\subsection{Quantum circuit for time evolution by differential operators}

Now, we consider a hyperbolic partial differential equation such that it can be reduced to the Schr\"{o}dinger equation $d \ket{u(t)}/dt = -i \mathcal{H} \ket{u(t)}$, where the Hamiltonian $\mathcal{H}$ consists of the difference operators, as exemplified in Eqs.~\eqref{eq:hamiltonian_advect} and \eqref{eq:hamiltonian_wave} in Section~\ref{sec:to_schroudinger}. 
Note that, in addition to those motivating examples, a wide class of partial differential equations also falls into our target based on the technique proposed in Ref.~\cite{jin2023aquantum, jin2023bquantum}. 
Our goal is to provide an efficient method for implementing the time evolution operator $\exp (-i \mathcal{H} \tau)$ for a time increment $\tau$, on a circuit of qubit-based quantum devices.

\begin{figure*}[t]
    \centering
    \includegraphics[width=0.9\textwidth]{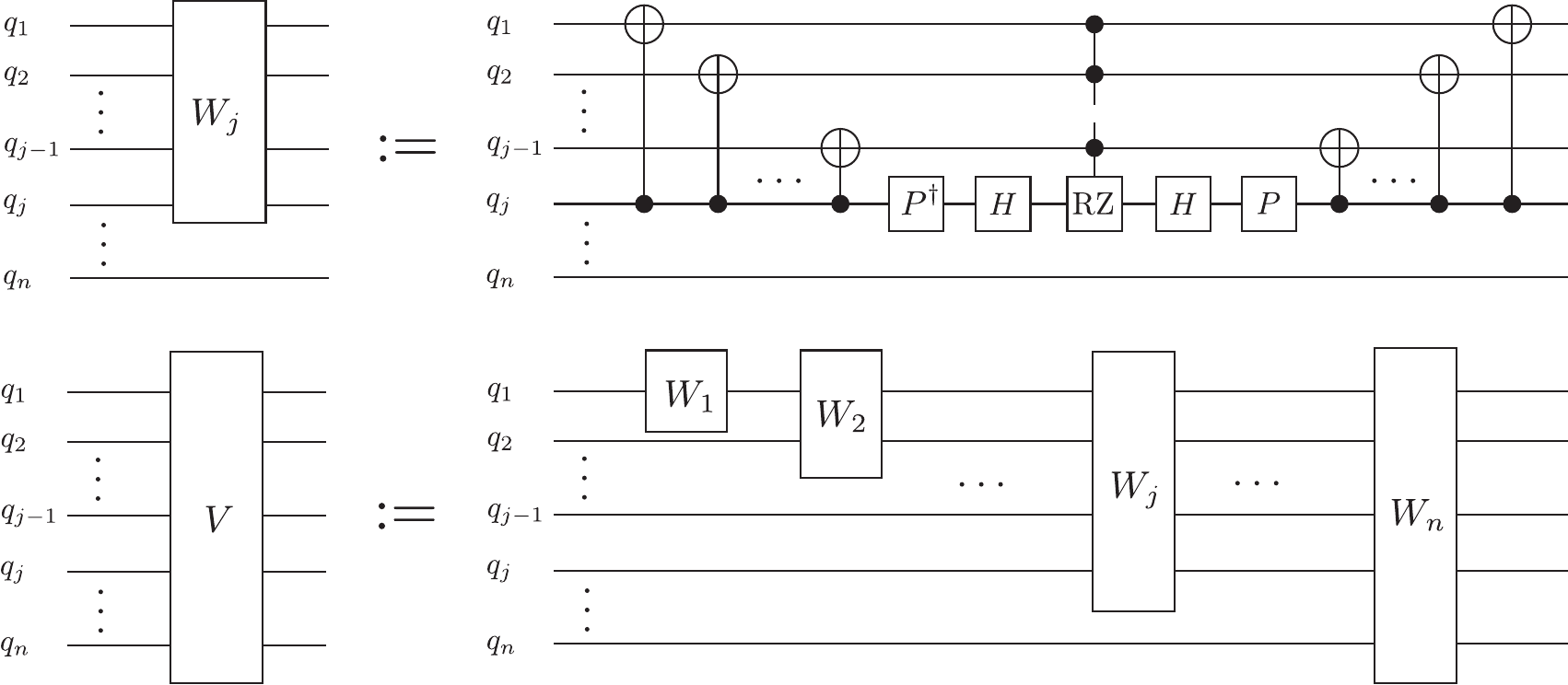}
    \caption{Quantum circuit to implement the time evolution operator $V$, where $q_j$ represents the $j$-th qubit.}
    \label{fig:circuit}
\end{figure*}

Let us consider the following Hamiltonian that contains the one-dimensional spatial difference operator: 
\begin{align}
    \mathcal{H} = \gamma \sum_{j=1}^n \left( e^{i\lambda} s^{-}_j + e^{-i\lambda} s^{+}_j  \right),
    \label{eq:hamiltonian_diffop}
\end{align}
where $\gamma \in \mathbb{R}$ is a scale parameter and $\lambda \in \mathbb{R}$ is a phase parameter.
This Hamiltonian can represent the essential part of the one-dimensional Hamiltonian for the advection equation given in Eq.~\eqref{eq:hamiltonian_advect}. 
Also, the following procedure for constructing quantum circuits and the scaling of the circuit complexity, based on the Hamiltonian \eqref{eq:hamiltonian_diffop}, is applicable to the Hamiltonian for the wave equation given in Eq.~\eqref{eq:hamiltonian_wave}. 
Since $\sigma_{01}=(X+iY)/2$ and $\sigma_{10}=(X-iY)/2$, where $X$ and $Y$ are Pauli matrices, we can naively represent each term of the Hamiltonian \eqref{eq:hamiltonian_diffop} by Pauli strings.
However, such representation of the Hamiltonian yields the exponentially large number of terms because $s_n^{-}$ and $s_n^{+}$ are global.
Here, we use the Bell basis instead of the Pauli matrices, which can efficiently diagonalize each term of the Hamiltonian~\eqref{eq:hamiltonian_diffop}, as follows:
\begin{widetext}
\begin{align} \label{eq:shift_general}
    e^{i\lambda} s^{-}_j + e^{-i \lambda} s^{+}_j &= e^{i\lambda} I^{\otimes (n-j)} \otimes \sigma_{01} \otimes \sigma_{10}^{\otimes (j-1)} + e^{-i \lambda} I^{\otimes (n-j)} \otimes \sigma_{10} \otimes \sigma_{01}^{\otimes (j-1)} \nonumber \\
    &= I^{\otimes (n-j)} \otimes \left( e^{i\lambda} \ket{0} \! \ket{1}^{\otimes (j-1)} \! \bra{1} \! \bra{0}^{\otimes (j-1)} + e^{-i\lambda} \ket{1} \! \ket{0}^{\otimes (j-1)} \! \bra{0} \! \bra{1}^{\otimes (j-1)} \right) \nonumber \\
    &= I^{\otimes (n-j)} \otimes \frac{\ket{0} \! \ket{1}^{\otimes (j-1)} + e^{-i \lambda} \ket{1} \! \ket{0}^{\otimes (j-1)}}{\sqrt{2}} \frac{\bra{0} \! \bra{1}^{\otimes (j-1)} + e^{i\lambda}  \bra{1} \! \bra{0}^{\otimes (j-1)}}{\sqrt{2}} \nonumber \\
    & \quad - I^{\otimes (n-j)} \otimes \frac{\ket{0} \! \ket{1}^{\otimes (j-1)} - e^{-i\lambda} \ket{1} \! \ket{0}^{\otimes (j-1)}}{\sqrt{2}} \frac{\bra{0} \! \bra{1}^{\otimes (j-1)} - e^{i\lambda} \bra{1} \! \bra{0}^{\otimes (j-1)}}{\sqrt{2}} \nonumber \\
    &= I^{ \otimes (n-j)} \otimes U_j(-\lambda) \left( Z \otimes \ket{1} \! \bra{1}^{\otimes (j-1)}\right) U_j(-\lambda)^\dagger,
\end{align}
\end{widetext}
where $Z$ is the single-qubit Z gate. 
Here we call $(\ket{0} \! \ket{1}^{\otimes (j-1)} \pm e^{-i \lambda} \ket{1} \! \ket{0}^{\otimes (j-1)}) / \sqrt{2}$ the Bell basis; the Bell basis is the key for decomposing the Hamiltonian into a sum of polynomial number of terms. 
Also, $U_j$ is the unitary matrix so that $U_j (-\lambda)\ket{0}\ket{1}^{\otimes (j-1)}=(\ket{0}\ket{1}^{\otimes (j-1)}+e^{-i\lambda}\ket{1}\ket{0}^{\otimes (j-1)})/\sqrt{2}$ defined as
\begin{align} \label{eq:bell_basis}
    U_j(\lambda) := \left( \prod_{m=1}^{j-1} \cnot^j_m \right) P_j(\lambda) H_j,
\end{align}
where $H_j$ is the Hadamard gate acting on the $j$-th qubit, $P_j(\lambda)$ is the Phase gate acting on the $j$-th qubit as
\begin{align}
    P_j(\lambda) := \begin{pmatrix}
        1 & 0 \\
        0 & e^{i\lambda}
    \end{pmatrix},
\end{align}
and $\cnot^j_m$ is the CNOT gate acting on the $m$-th qubit controlled by the $j$-th qubit.
Note that we herein use the little endian.
Applying the first-order Lie-Trotter-Suzuki decomposition, we can approximate the time evolution operator $\exp (-i \mathcal{H} \tau)$, as follows:
\begin{align} \label{eq:exp_iHt}
    &\exp \left(-i \mathcal{H} \tau \right) = \exp \left( -i \gamma \tau \sum_{j=1}^n (e^{i\lambda} s^{-}_j + e^{-i\lambda} s^{+}_j) \right) \nonumber \\
    &\approx \prod_{j=1}^{n} \exp \left( -i \gamma \tau I^{ \otimes (n-j)} \right. \nonumber \\
    & \qquad \qquad \left. \otimes U_j(-\lambda) \left( Z \otimes \ket{1} \! \bra{1}^{\otimes (j-1)}\right) U_j(-\lambda)^\dagger \right) \nonumber \\
    & = \prod_{j=1}^{n} I^{ \otimes (n-j)} \otimes U_j(-\lambda) \crz^{1, \dots, j-1}_j \left( 2\gamma \tau \right) U_j(-\lambda)^\dagger \nonumber \\
    &= \prod_{j=1}^{n} W_j\left( \gamma \tau, \lambda \right),
\end{align}
where $\crz^{1, \dots, j-1}_j(\theta):=\exp(-i\theta Z_j/2)\otimes\ket{1}\!\bra{1}^{\otimes(j-1)} + I \otimes (I^{\otimes (j-1)} - \ket{1}\!\bra{1}^{\otimes(j-1)})$ is the multi-controlled RZ gate acting on the $j$-th qubit controlled by $1, \dots, (j-1)$-th qubits. 
Note that from the second to the third line, $\exp(I\otimes A)=I\otimes \exp(A)$ and $\exp(UAU^\dagger)=U\exp(A)U^\dagger$ were used.
The unitary matrix $W_j$ is defined as
\begin{align}
    &W_j\left( \gamma \tau, \lambda \right) \nonumber \\
    &:= I^{ \otimes (n-j)} \otimes U_j(-\lambda) \crz^{1, \dots, j-1}_j \left( 2\gamma \tau \right) U_j(-\lambda)^\dagger.
\end{align}
Then, let $V\left( \gamma \tau, \lambda \right)$ denote the approximated time evolution operator as
\begin{align} \label{eq:trotter_general}
    V\left( \gamma \tau, \lambda \right) := \prod_{j=1}^{n} W_j\left( \gamma \tau, \lambda \right) \approx \exp \left(-i \mathcal{H} \tau \right).
\end{align}

We now have a concrete circuit implementation of this approximating unitary matrix $V$, as shown in Fig.~\ref{fig:circuit}. 
Note that such explicit form of implementation has not been reported in the previous proposals~\cite{jin2023aquantum, jin2023bquantum}.
Moreover, thanks to the explicit form of $V$, we can have a detailed evaluation on the approximation error between $\exp(-i \mathcal{H} \tau)$ and $V$ in the sense of operator norm. 
Together with the explicit circuit construction, the following lemma gives the approximation error.

\begin{lem} \label{lem:diff_circ}
    Consider the Schr\"{o}dinger equation 
    $d \ket{u(t)}/dt = -i \mathcal{H} \ket{u(t)}$ such that
    the Hamiltonian $\mathcal{H}$ is given by Eq.~\eqref{eq:hamiltonian_diffop}.
    The time evolution operator $\exp(-i \mathcal{H} \tau)$ with the time increment $\tau$ can be approximated by the unitary $V$ in Eq.~\eqref{eq:trotter_general}, and its explicit circuit implementation is shown in Fig.~\ref{fig:circuit}. 
    Moreover, the approximation error in the sense of the operator norm is upper bounded as 
    \begin{align}
        \left\Vert \exp(-i \mathcal{H} \tau) - V( \gamma \tau, \lambda) \right\Vert &\leq \frac{\gamma^2 \tau^2 (n-1)}{2}.
    \end{align}
\end{lem}
\begin{proof}
    We use the fact that the approximation error of the Lie-Trotter-Suzuki decomposition is upper bounded by the sum of operator norm of commutators of all terms contained in the Hamiltonian \cite[Proposition 9]{childs2021theory}. 
    Hence our task is to evaluate the commutators of all terms of the Hamiltonian \eqref{eq:hamiltonian_diffop}, as follows; the detailed calculation is given in Appendix~\ref{sec:proof_lemma}. 
    For $n>j>j'>1$, the terms of the shift operators $e^{i\lambda} s^{-}_j + e^{-i\lambda}s^{+}_j$ and $e^{i\lambda} s^{-}_{j'} + e^{-i\lambda} s^{+}_{j'}$ commute as
    \begin{align}
        \left[ e^{i\lambda} s^{-}_j + e^{-i\lambda} s^{+}_j, e^{i\lambda} s^{-}_{j'} + e^{-i\lambda} s^{+}_{j'} \right] &= 0.
    \end{align}
    For $j>j'=1$, we obtain
    \begin{align}
        \left\Vert \left[ e^{i\lambda} s^{-}_j + e^{-i\lambda} s^{+}_j, e^{i\lambda} s^{-}_1 + e^{-i\lambda} s^{+}_1 \right] \right\Vert &= 1.
    \end{align}
    Thus, the terms of the Hamiltonian, $\gamma(e^{i\lambda} s^{-}_j + e^{-i\lambda} s^{+}_j)$, can be grouped into those for $j>1$ and $j=1$.
    Since the unitary does not change the operator norm, the Trotter error is upper bounded using the result of Ref.~\cite[Proposition 9]{childs2021theory}, as follows: %
    \begin{align}
        &\left\Vert \exp(-i \mathcal{H} \tau) - V( \gamma \tau, \lambda ) \right\Vert \nonumber \\
        &\leq \frac{\gamma^2 \tau^2}{2} \left\Vert \left[ \sum_{j=2}^n \left( e^{i\lambda} s^{-}_j + e^{-i\lambda} s^{+}_j \right), e^{i\lambda} s^{-}_1 + e^{-i\lambda} s^{+}_1 \right] \right\Vert \nonumber \\
        &\leq \frac{\gamma^2 \tau^2}{2} \sum_{j=2}^n \left\Vert \left[  e^{i\lambda} s^{-}_j + e^{-i\lambda} s^{+}_j, e^{i\lambda} s^{-}_1 + e^{-i\lambda} s^{+}_1 \right] \right\Vert \nonumber \\
        &= \frac{\gamma^2 \tau^2 (n-1)}{2}
    \end{align}
\end{proof}

Since $W_{j>2}$ and $W_{j'>2}$ commute, we can easily obtain the second-order formula as
\begin{align} \label{eq:trotter_second}
    V^{(2)}(\gamma \tau, \lambda):= W_1 (\gamma \tau / 2, \lambda) V W_1 (-\gamma \tau / 2, \lambda).
\end{align}
The following lemma gives the approximation error
between $\exp(-i \mathcal{H} \tau)$ and $V^{(2)}$ in the sense of operator norm.

\begin{lem} \label{lem:diff_circ_second}
    Consider the Schr\"{o}dinger equation $d \ket{u(t)}/dt = -i \mathcal{H} \ket{u(t)}$ such that the Hamiltonian $\mathcal{H}$ is given by Eq.~\eqref{eq:hamiltonian_diffop}.
    The time evolution operator $\exp(-i \mathcal{H} \tau)$ with the time increment $\tau$ can be approximated by the unitary $V^{(2)}$ in Eq.~\eqref{eq:trotter_second}. 
    The approximation error in the sense of the operator norm is upper bounded as 
    \begin{align}
        \left\Vert \exp(-i \mathcal{H} \tau) - V^{(2)} ( \gamma \tau, \lambda) \right\Vert &\leq \frac{\gamma^3 \tau^3}{6} (2n-3).
    \end{align}
\end{lem}
\begin{proof}
    The detailed proof is given in Appendix~\ref{sec:proof_lemma_second}. 
    Here, we provide a scketch of the proof.
    Let us group the terms of the Hamiltonian $\mathcal{H}$ in Eq.~\eqref{eq:hamiltonian_diffop} into $\mathcal{H}_1 := \gamma e^{i\lambda} + s_1^{-} + e^{-i\lambda}s_1^{+}$ and $H_2 := \gamma \sum_{j=2}^n (e^{i\lambda} + s_j^{-} + e^{-i\lambda}s_j^{+})$ such that $\mathcal{H} = \mathcal{H}_1 + \mathcal{H}_2$.
    Since all terms in $\mathcal{H}_2$ commute each other, the error of the second-order Suzuki formula is upper bounded~\cite[Proposition 9]{childs2021theory} by
    \begin{align}
        &\left\Vert \exp(-i \mathcal{H} \tau) - V^{(2)}( \gamma \tau, \lambda ) \right\Vert \nonumber \\
        &\quad \leq \frac{\tau^3}{12} \left\Vert \left[ \mathcal{H}_2, \left[ \mathcal{H}_2, \mathcal{H}_1 \right] \right] \right\Vert + \frac{\tau^3}{24} \left\Vert \left[ \mathcal{H}_1, \left[ \mathcal{H}_1, \mathcal{H}_2 \right] \right] \right\Vert.
    \end{align}
    By evaluating the commutators, we obtain
    \begin{align}
        \left\Vert \exp(-i \mathcal{H} \tau) - V^{(2)}( \gamma \tau, \lambda ) \right\Vert &\leq \frac{\gamma^3 \tau^3}{6} (2n-3).
    \end{align}   
\end{proof}

We also provide the following lemma about the gate counts of the unitary $V$.
\begin{lem} \label{lem:gate_count}
    The approximated time evolution operators $V$ in Fig.~\ref{fig:circuit} and $V^{(2)}$ can be implemented using single-qubit gates and at most $9n^2-33n+34$ CNOT gates for $n \geq 3$.
\end{lem}
\begin{proof}
    As shown in Fig.~\ref{fig:circuit}, the non-local gates included in the operator $W_j$ are a multi-controlled Rz gate for $j \geq 3$ (a controlled Rz gate for $j=2$) and totally $2(j-1)$ CNOT gates in the unitary $U_j$ and $U_j^\dagger$ for $j \geq 2$. 
    It is known that the multi-controlled RZ gate with $(j-1)$ control qubits can be decomposed into single-qubit gates and at most $16j-40$ CNOT gates~\cite[Theorem 3]{vale2023decomposition}.%
    Therefore, the number of CNOT gates required to implement the approximated time evolution operator $V$ is
    \begin{align}
        \sum_{j=3}^n (16j-40) + 2 + \sum_{j=2}^n 2(j-1) = 9n^2 -33n + 34.
    \end{align}
    Since operators $V$ and $V^{(2)}$ have the difference only in the single-qubit gate by definition in Eq.~\eqref{eq:trotter_second}, the number of CNOT gates in $V^{(2)}$ is the same as that in $V$.
\end{proof}

We now extend the above discussion to the operators acting on a $d$-dimensional domain $\Omega$.
\begin{lem} \label{lem:diff_circ_d-dim}
    Let the Hamiltonian $\mathcal{H}$ consist of finite difference operators for a $dn$-qubit system as
    \begin{align} \label{eq:hamiltonian_d-dim}
        \mathcal{H} = \gamma \sum_{\alpha=1}^d \sum_{j=1}^n \eta_\alpha \left( e^{i\lambda_\alpha} (s^{-}_j)_\alpha + e^{-i\lambda_\alpha} (s^{+}_j)_\alpha  \right),
    \end{align}
    where $\gamma \in \mathbb{R}$ is a scale parameter, $\lambda_\alpha \in \mathbb{R}$ is the phase parameter and 
    \begin{align}
        (s^\mu_j)_\alpha = I^{\otimes (\alpha-1)n} \otimes s^\mu_j \otimes I^{\otimes (d-\alpha)n},
    \end{align}
    for $\mu \in \{ -, + \}$.
    The time evolution operator $\exp(-i \mathcal{H} \tau)$ with the time increment $\tau$ can be approximated by the unitary $\bigotimes_{\alpha=1}^d V(\gamma \eta_\alpha \tau, \lambda_\alpha)$. 
    The approximation error is upper bounded in the sense of the operator norm as
    \begin{align}
        &\left\Vert \exp \left(-i \mathcal{H} \tau \right) - \bigotimes_{\alpha=1}^d V(\gamma \eta_\alpha \tau, \lambda_\alpha) \right\Vert \nonumber \\
        &\qquad \qquad \qquad \qquad \leq \frac{\gamma^2 \tau^2 (n-1) \sum_{\alpha=1}^d \eta_\alpha^2}{2}.
    \end{align}
\end{lem}
\begin{proof}
    Here we sketch the proof; the detail is given in  Appendix~\ref{sec:proof_lemma}.
    From Eq.~\eqref{eq:shift_general}, we obtain
    \begin{align}
        &\exp \left(-i \mathcal{H} \tau \right) \nonumber \\
        &= \exp \left( -i \gamma \tau \sum_{\alpha=1}^d \sum_{j=1}^n \eta_\alpha (e^{i\lambda_\alpha} (s^{-}_j)_\alpha + e^{-i\lambda_\alpha} (s^{+}_j)_\alpha) \right) \nonumber \\
        &\approx \prod_{\alpha=1}^d I^{\otimes (\alpha - 1)n } \otimes V(\gamma \eta_\alpha \tau, \lambda_\alpha) \otimes I^{\otimes (d - \alpha)n} \nonumber \\
        &= \bigotimes_{\alpha=1}^d V\left( \gamma \tau \eta_\alpha, \lambda_\alpha \right),
    \end{align}
    where $V(\gamma \eta_\alpha \tau, \lambda_\alpha)$ is given in Eq.~\eqref{eq:trotter_general} and represents the time evolution operator for the spatial dimension in the $x_\alpha$ direction.
    The approximation error of the Lie-Trotter-Suzuki decomposition is upper bounded by the operator norm of commutators of Hamiltonian~\cite[Proposition 9]{childs2021theory}.
    Together with the fact the terms of the Hamiltonian $e^{i\lambda_\alpha} (s^{-}_j)_\alpha + e^{-i\lambda_\alpha} (s^{+}_j)_\alpha$ and $e^{i\lambda_{\alpha'}} (s^{-}_j)_{\alpha'} + e^{-i\lambda_{\alpha'}} (s^{+}_j)_{\alpha'}$ commute for $\alpha \neq \alpha'$ and the discussion in the proof of Lemma~\ref{lem:diff_circ}, we obtain
    \begin{align}
        &\left\Vert \exp \left(-i \mathcal{H} \tau \right) - \bigotimes_{\alpha=1}^d V\left( \gamma \tau \eta_\alpha, \lambda_\alpha \right) \right\Vert \nonumber \\
        &\qquad \qquad \qquad \qquad \leq \frac{\gamma^2 \tau^2 (n-1) \sum_{\alpha=1}^d \eta_\alpha^2}{2}.
    \end{align}
\end{proof}

We can easily extend this lemma to the second-order formula as follows.
\begin{lem} \label{lem:diff_circ_d-dim_second}
    Let us consider the Hamiltonian $\mathcal{H}$ for a $dn$-qubit system in Eq.~\eqref{eq:hamiltonian_d-dim}.
    The time evolution operator $\exp(-i \mathcal{H} \tau)$ with the time increment $\tau$ can be approximated by the unitary $\bigotimes_{\alpha=1}^d V^{(2)}(\gamma \eta_\alpha \tau, \lambda_\alpha)$. 
    The approximation error is upper bounded in the sense of the operator norm as
    \begin{align}
        &\left\Vert \exp \left(-i \mathcal{H} \tau \right) - \bigotimes_{\alpha=1}^d V^{(2)}(\gamma \eta_\alpha \tau, \lambda_\alpha) \right\Vert \nonumber \\
        &\qquad \qquad \qquad \qquad \leq \frac{\gamma^3 \tau^3 (2n-3) \sum_{\alpha=1}^d \eta_\alpha^3}{6}.
    \end{align}
\end{lem}

\begin{proof}
    From Eq.~\eqref{eq:shift_general} and the fact that $W_{j>2}$ and $W_{j'>2}$ commute, we obtain
    \begin{align}
        &\exp \left(-i \mathcal{H} \tau \right) \nonumber \\
        &= \exp \left( -i \gamma \tau \sum_{\alpha=1}^d \sum_{j=1}^n \eta_\alpha (e^{i\lambda_\alpha} (s^{-}_j)_\alpha + e^{-i\lambda_\alpha} (s^{+}_j)_\alpha) \right) \nonumber \\
        &\approx \prod_{\alpha=1}^d I^{\otimes (\alpha - 1)n } \otimes V^{(2)}(\gamma \eta_\alpha \tau, \lambda_\alpha) \otimes I^{\otimes (d - \alpha)n} \nonumber \\
        &= \bigotimes_{\alpha=1}^d V^{(2)}\left( \gamma \tau \eta_\alpha, \lambda_\alpha \right),
    \end{align}
    where $V^{(2)}(\gamma \eta_\alpha \tau, \lambda_\alpha)$ is given in Eq.~\eqref{eq:trotter_second}.
    The approximation error of the Lie-Trotter-Suzuki decomposition is upper bounded by the operator norm of commutators of Hamiltonian~\cite[Proposition 9]{childs2021theory}.
    Together with the fact the terms of the Hamiltonian $e^{i\lambda_\alpha} (s^{-}_j)_\alpha + e^{-i\lambda_\alpha} (s^{+}_j)_\alpha$ and $e^{i\lambda_{\alpha'}} (s^{-}_j)_{\alpha'} + e^{-i\lambda_{\alpha'}} (s^{+}_j)_{\alpha'}$ commute for $\alpha \neq \alpha'$ and the discussion in the proof of Lemma~\ref{lem:diff_circ_second}, we obtain
    \begin{align}
        &\left\Vert \exp \left(-i \mathcal{H} \tau \right) - \bigotimes_{\alpha=1}^d V\left( \gamma \tau \eta_\alpha, \lambda_\alpha \right) \right\Vert \nonumber \\
        &\qquad \qquad \qquad \qquad \leq \frac{\gamma^3 \tau^3 (2n-3) \sum_{\alpha=1}^d \eta_\alpha^3}{6}.
    \end{align}
\end{proof}

We also have the following lemma about the gate counts of the unitary $\bigotimes_{\alpha=1}^d V\left( \gamma \tau \eta_\alpha, \lambda_\alpha \right)$ and $\bigotimes_{\alpha=1}^d V^{(2)}\left( \gamma \tau \eta_\alpha, \lambda_\alpha \right)$.
\begin{lem} \label{lem:gate_count_d-dim}
    The approximated time evolution operators $\bigotimes_{\alpha=1}^d V\left( \gamma \tau \eta_\alpha, \lambda_\alpha \right)$ and $\bigotimes_{\alpha=1}^d V^{(2)}\left( \gamma \tau \eta_\alpha, \lambda_\alpha \right)$ with $V$ and $V^{(2)}$ given in Eqs.~\eqref{eq:trotter_general} and \eqref{eq:trotter_second}, respectively, can be implemented using single-qubit gates and at most $d(9n^2 -33n + 34)$ CNOT gates for $n \geq 3$.
\end{lem}
\begin{proof}
    Based on the same discussion in the proof of Lemma~\ref{lem:gate_count}, the number of CNOT gates included in the approximated time evolution operators $V$ and $V^{(2)}$ are $9n^2 -33n + 34$.
    Therefore, the approximated time evolution operators $\bigotimes_{\alpha=1}^d V( \gamma \tau \eta_\alpha, \lambda_\alpha)$ and $\bigotimes_{\alpha=1}^d V^{(2)}( \gamma \tau \eta_\alpha, \lambda_\alpha)$ can be implemented by single-qubits gates and at most $d(9n^2 -33n + 34)$ CNOT gates.
\end{proof}

To simulate the Hamiltonian dynamics over the total time $T$, it suffices to divide the total time $T$ into $r(:=T/\tau)$ intervals so that the approximation error occurred in the time interval $\tau$ could be small enough to be acceptable. 

We remark that the essential part of the Hamiltonian \eqref{eq:hamiltonian_advect} for the advection equation falls into the Hamiltonian \eqref{eq:hamiltonian_d-dim} by setting $\gamma=1/(2l)$, $\eta_\alpha=v_\alpha$ and $\lambda_\alpha=-\pi/2$. 
Appendix~\ref{sec:impl_advect} gives the procedure for constructing the quantum circuit and its Trotter error, for the full Hamiltonian of the advection equation. 
Also for the case of wave equation, although the Hamiltonian in Eq.~\eqref{eq:hamiltonian_wave} does not fit into that in Lemma~\ref{lem:diff_circ}, we can easily obtain the quantum circuit for Hamiltonian simulation as discussed in Appendix~\ref{sec:impl_wave}.

\subsection{Space and time complexities} \label{sec:complexity}

Together with the quantization of the field in Eq.~\eqref{eq:u_d_dim}, the explicit circuit construction and Lemmas~\ref{lem:diff_circ_d-dim} and \ref{lem:gate_count_d-dim}, we now provide the following theorem giving the space and time complexities of our Hamiltonian simulation, as our main result.
\begin{thm} \label{thm:complexity}
    Let $\mathcal{H}$ be the Hamiltonian as defined in Eq.~\eqref{eq:hamiltonian_d-dim}, i.e., $\mathcal{H}=\gamma \sum_{\alpha=1}^d \sum_{j=1}^n \eta_\alpha ( e^{i\lambda_\alpha} (s^{-}_j)_\alpha + e^{-i\lambda_\alpha} (s^{+}_j)_\alpha)$. 
    The time evolution operator $\exp(-i \mathcal{H} T)$ up to time $T$ is implementable on the $dn$-qubits system using the quantum circuits with $O(d n^3 \gamma^2 T^2 \sum_{\alpha=1}^d \eta_\alpha^2 / \varepsilon)$ non-local gates within the additive error $\varepsilon$.
    Furthermore, the leading term of the number of non-local gates is $9d n^3 \gamma^2 T^2 \sum_{\alpha=1}^d \eta_\alpha^2 / (2\varepsilon)$.
    The quantum circuit for $\exp(-i \mathcal{H} T)$ consists of the repetitive applications of the one time step unitary $V$ shown in Fig.~\ref{fig:circuit}.
\end{thm}
\begin{proof}
    Lemma~\ref{lem:diff_circ_d-dim} states that the additive error of the approximated time evolution operator $\bigotimes_{\alpha=1}^d V\left( \gamma \tau \eta_\alpha, \lambda_\alpha \right)$ scales $\gamma^2 \tau^2 (n-1) \sum_{\alpha=1}^d \eta_\alpha^2 / 2$ with the time increment $\tau$.
    To suppress the error of the simulation over the total time $T$ within a small value $\varepsilon$, it suffices to divide the total time $T$ into $r(:=T/\tau)$ intervals so that 
    \begin{align}
        \frac{\gamma^2 T^2 (n-1) \sum_{\alpha=1}^d \eta_\alpha^2}{2r^2} \leq \frac{\varepsilon}{r},
    \end{align}
    which is rearranged as
    \begin{align}
        r \geq \frac{\gamma^2 T^2 (n-1) \sum_{\alpha=1}^d \eta_\alpha^2}{2 \varepsilon}.
    \end{align}
    Since each Trotter step $\bigotimes_{\alpha=1}^d V\left( \gamma \tau \eta_\alpha, \lambda_\alpha \right)$ requires $d(9n^2-33n+34)$ CNOT gates by Lemma~\ref{lem:gate_count_d-dim}, Hamiltonian simulation up to time $T$ within the additive error $\varepsilon$ in the sense of the operator norm requires $rd(9n^2-33n+34) \geq d (9n^3 - 42n^2 +67n - 34) \gamma^2 T^2 \sum_{\alpha=1}^d \eta_\alpha^2 / (2 \varepsilon)$
    CNOT gates which scales $O(d n^3 \gamma^2 T^2 \sum_{\alpha=1}^d \eta_\alpha^2 / \varepsilon)$.
\end{proof}

As for the second-order formula, we also provide the following theorem giving the space and time complexities of our Hamiltonian simulation.

\begin{thm} \label{thm:complexity_second}
    Let $\mathcal{H}$ be the Hamiltonian as defined in Eq.~\eqref{eq:hamiltonian_d-dim}, i.e., $\mathcal{H}=\gamma \sum_{\alpha=1}^d \sum_{j=1}^n \eta_\alpha ( e^{i\lambda_\alpha} (s^{-}_j)_\alpha + e^{-i\lambda_\alpha} (s^{+}_j)_\alpha)$. 
    The time evolution operator $\exp(-i \mathcal{H} T)$ up to time $T$ is implementable on the $dn$-qubits system using the quantum circuits with $O(dn^{2.5} \gamma^{1.5} T^{1.5} (\sum_{\alpha=1}^d \eta_\alpha^3)^{0.5} / \varepsilon^{0.5})$ non-local gates within the additive error $\varepsilon$. 
    Furthermore, the leading term of the number of non-local gates is $(3 \sqrt{3} dn^{2.5} \gamma^{1.5} T^{1.5} (\sum_{\alpha=1}^d \eta_\alpha^3)^{0.5} / (\varepsilon)^{0.5})$.
    The quantum circuit for $\exp(-i \mathcal{H} T)$ consists of the repetitive applications of the one time step unitary $V^{(2)}$ in Eq.~{\eqref{eq:trotter_second}}.
\end{thm}
\begin{proof}
    Lemma~\ref{lem:diff_circ_d-dim_second} states that the additive error of the approximated time evolution operator $\bigotimes_{\alpha=1}^d V^{(2)}\left( \gamma \tau \eta_\alpha, \lambda_\alpha \right)$ scales $\gamma^3 \tau^3 (2n-3) \sum_{\alpha=1}^d \eta_\alpha^3 / 6$ with the time increment $\tau$.
    To suppress the error of the simulation over the total time $T$ within a small value $\varepsilon$, we divide the total time $T$ into $r(:=T/\tau)$ intervals so that 
    \begin{align}
        \frac{\gamma^3 T^3 (2n-3) \sum_{\alpha=1}^d \eta_\alpha^3}{6r^3} \leq \frac{\varepsilon}{r},
    \end{align}
    which is rearranged as
    \begin{align}
        r \geq \frac{\gamma^{1.5} T^{1.5} (2n-3)^{0.5} (\sum_{\alpha=1}^d \eta_\alpha^3)^{0.5}}{(6\varepsilon)^{0.5}}.
    \end{align}
    Since each Trotter step $\bigotimes_{\alpha=1}^d V^{(2)}\left( \gamma \tau \eta_\alpha, \lambda_\alpha \right)$ requires $d(9n^2-33n+34)$ CNOT gates by Lemma~\ref{lem:gate_count_d-dim}, Hamiltonian simulation up to time $T$ within the additive error $\varepsilon$ in the sense of the operator norm requires $rd(9n^2-33n+34)$ CNOT gates whose leading term is $(3 \sqrt{3} dn^{2.5} \gamma^{1.5} T^{1.5} (\sum_{\alpha=1}^d \eta_\alpha^3)^{0.5} / (\varepsilon)^{0.5})$.
    This scales $O(dn^{2.5} \gamma^{1.5} T^{1.5} (\sum_{\alpha=1}^d \eta_\alpha^3)^{0.5} / \varepsilon^{0.5})$.
\end{proof}

Note that since $V^r = W_1(-\gamma \tau, \lambda) (V^{(2)})^r W_1(\gamma \tau, \lambda)$, the first-order formula works similarly to the second-order formula and can exhibit the similar bound to the second-order formula in a practical sense~\cite{layden2022first}.

\begin{rem} \label{rem:classical_complexity}
    Classical implementation requires $O(2^{dn})$ memories to store the discretized scalar field $\bm{u}$.
    Since the finite difference operators, such as $D^\mu_\mathrm{B}$ for $\mu \in \{-, +, \pm, \Delta \}$ and $\mathrm{B} \in \{ \mathrm{D}, \mathrm{N}, \mathrm{P} \}$ in Eqs.~\eqref{eq:forward diff operator}, \eqref{eq:backward diff operator}, \eqref{eq:central diff operator} and \eqref{eq:laplacian operator}, can be represented by sparse matrices of the size $2^{dn} \times 2^{dn}$ with the sparsity denoted by $s$, the application of the operator $D^\mu_\mathrm{B}$ to the vector $\bm{u}$ requires $O(s 2^{dn})$ arithmetic operations.
    Here, $s=2$ for $\mu \in \{-, +, \pm\}$, i.e., the first-order derivative and $s=3$ for $\mu =\Delta$, i.e., the Laplacian in Eq.~\eqref{eq:laplacian operator}.
    Now, let us consider classical simulation using the forward Euler scheme, which has the additive error bounded by $O(\tau^2)$ where $\tau$ is the time increment.
    Then, classical simulation up to time $T$ within the additive error $\varepsilon$ requires $O(T^2 / \varepsilon)$ steps, which results in $O(s2^{dn} (T^2 / \varepsilon))$ arithmetic operations.
    When using the second-order formula for the time integration, the complexity will be improved to $O(s2^{dn} (T^{1.5} / \varepsilon^{0.5}))$.
    In addition, the Courant-Friedrichs-Lewy (CFL) condition~\cite{de2013courant} requires $T/r = O(l)$.
    As a result, classical simulation up to time $T$ within the additive error $\varepsilon$ requires $O(s2^{dn} (T^2 / \varepsilon + T / l))$ or $O(s2^{dn} (T^{1.5} / \varepsilon^{0.5} + T / l))$ arithmetic operations.
\end{rem}
\begin{rem} \label{rem:comparison_complexity}
    Focusing on the spatial dimension $d$, the number of nodes in each dimension $2^n$, the total simulation time $T$, and the additive error of the simulation, $\varepsilon$, our Hamiltonian simulation requires $O(d n^3 T^2 / \varepsilon)$ non-local quantum gate operations under the assumption of $\gamma=O(1)$ and $\eta_\alpha=O(1)$ while classical approaches require $O(2^{dn} T^2 / \varepsilon)$ arithmetic operations.
    Considering the relationship of $n=O(\log(L/l))$ and assuming $\gamma = O(1/l)$, which is the case of the advection and wave equations, our Hamiltonian simulation requires $O(d T^2 \log(L/l)^3 / (l^2 \varepsilon))$ or $O(d T^{1.5} \log(L/l)^{2.5} / (l^{1.5} \varepsilon^{0.5}))$ under the assumption of $\eta_\alpha=O(1)$ while classical approaches require $O(L^d (T^2 / (l^d \varepsilon) + T / l^{d+1}))$ or $O(L^d (T^{1.5} / (l^d \varepsilon^{0.5}) + T / l^{d+1}))$.
\end{rem}

These theorem and remarks suggest that our method has the potential for exponential speedup with respect to the number of nodes on the lattice $N (= 2^n)$ and the size of the domain $L$ when simulating the classical dynamics governed by differential operators, if the operation of quantum gates can be performed as fast as the arithmetic operations in classical computers.
Remark~\ref{rem:comparison_complexity} also implies our method will exhibit the polynomial speedup with respect to the interval of the lattice $l$ when $d \geq 2$.

Note that the time complexity of our algorithm is polynomial with respect to the additive error $\varepsilon$ while the state-of-the-art algorithms~\cite{low2019hamiltonian, martyn2021grand} give the time complexity of $\mathrm{poly}(\log(1/\varepsilon))$.
Such algorithms rely on the oracle access to the block encoding of the Hamiltonian, while our current study focuses on deriving the explicit quantum circuit for Hamiltonian simulation.
We would like to conduct our future work to derive the quantum circuit for the block encoding-based Hamiltonian simulation, providing the comparison of the gate counts of the Trotter-based Hamiltonian simulation and the block encoding-based one including the constant factor in our future work.

Finally, we remark that it is impractical to access all components of $\ket{u(t)}$ or $\ket{\psi(t)}$ for the advection and the wave equation, respectively, because it requires $O(2^{dn})$ measurements.
Hence, the proposed method should be used in a situation when only some characteristic quantities about the solution are of interest; typically, such quantity is represented by $\braket{u(t) | \mathcal{O} | u(t)}$ or $\braket{\psi(t) | \mathcal{O} | \psi(t)}$ for an observable $\mathcal{O}$.
As we discuss later, one example of such observables is the kinetic energy of the system governed by the wave equation:
\begin{align}
    \mathcal{O} = \frac{1}{2} (Z + I) \otimes I^{\otimes n},
\end{align}
which leads to $\braket{\psi(t) | \mathcal{O} | \psi(t)} = \sum_x | \partial u (t, x) / \partial t |^2$.
The power spectra of the system is also one of the possible observables~\cite{miyamoto2024quantum}.
We can use various well established methods for such estimation of observables~\cite{knill2007optimal, alase2022tight}.
We would like to construct meaningful observables in more detail for specific applications and discuss efficient estimation of their expected values in our future work.

\section{Numerical and experimental results} \label{sec:result}

In this section, we provide several results of numerical experiments to demonstrate the validity of our proposed method.
We used Qiskit 0.45~\cite{Qiskit}, the open-source toolkit for quantum computation to implement quantum circuits.

\subsection{Advection equation}
\begin{figure*}[t]
    \centering
    \includegraphics[width=0.8\textwidth]{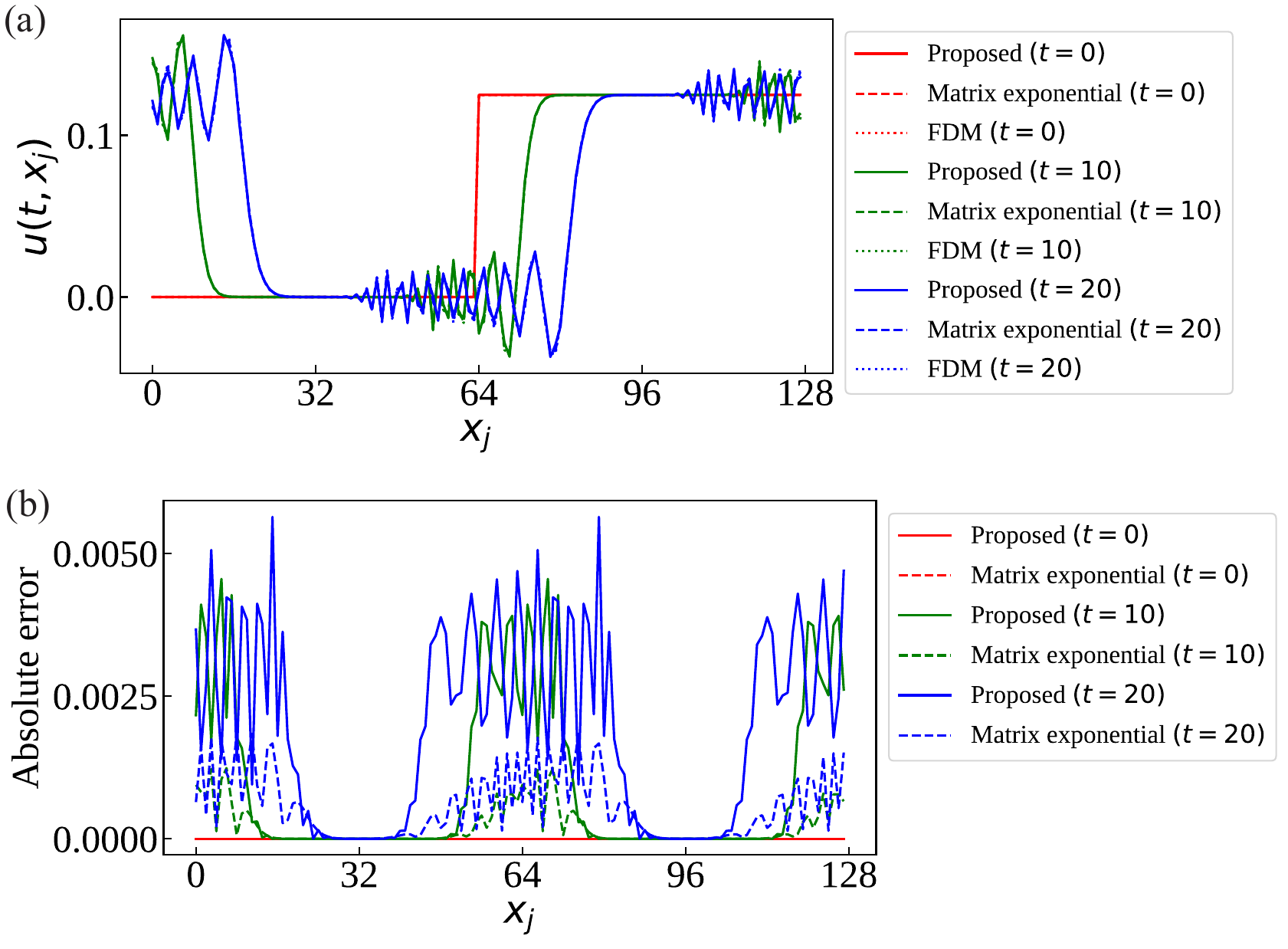}
    \caption{Simulation results of the advection equation with periodic BC in one dimension. (a) Solid lines represent the amplitudes of the quantum state $\ket{u(t)}$ generated by the proposed circuits. Dashed lines also represent the amplitudes of the quantum state $\ket{u(t)}$, but it was directly computed by applying the matrix exponential operator $e^{-i \mathcal{H} t}$ to $\ket{u(0)}$ via matrix-vector calculation. Dotted lines represent the solutions calculated by the finite difference method (FDM). 
    These three lines are overlapped very well in almost all the spatial region.
    (b) Absolute errors of the amplitudes of the quantum state $\ket{u(t)}$ by the proposed circuits (solide lines) and those applying the matrix exponential operator $e^{-i \mathcal{H} t}$ to $\ket{u(0)}$ via matrix-vector calculation (dashed lines), to the solutions calculated by FDM, respectively.
    }
    \label{fig:advect_1d}
\end{figure*}

We first performed Hamiltonian simulation for solving the advection equation with periodic BC, described in Section~\ref{sec:advect_sch}. 
Here, we compare the simulation result of our proposed method with the following two; 
the first one is the results calculated by directly applying the matrix exponential operator $e^{-i \mathcal{H} \tau}$ to $\ket{u(0)}$ by matrix-vector calculation, and the second one is the result calculated by the fully classical finite difference method (FDM) with the central differencing scheme given in Eq.~\eqref{eq:central diff operator}.

Figure~\ref{fig:advect_1d} illustrates the simulation results of the advection equation with periodic BC in one dimension.
We set $n=7$, $\tau=0.1$, $T=20$, $l=1$, and $v_1=1$ to formulate the problem and used \texttt{statevector\_simulator} to run the quantum circuits. 
Also, we set the time increment parameter to $0.01$ for the FDM.
As an initial state, we prepare
\begin{align}
    \ket{u(0)} = \frac{1}{\sqrt{2^{n-1}}} \ket{1} \otimes \sum_{j=0}^{2^{n-1}-1} \ket{j}.
\end{align}
Figure~\ref{fig:advect_1d} also include this initial state as the amplitude at $t=0$ to visualize the dynamics from the initial state.
Solid lines in Fig.~\ref{fig:advect_1d}(a) represent the amplitudes of the quantum state $\ket{u(t)}$ generated by the proposed circuits, at time $t=0,10,20$; note that in this case the probability amplitudes of $\ket{u(t)}$ are always real.
Dashed lines represent the amplitudes of the quantum state $\ket{u(t)}$ which is directly generated by applying the matrix exponential operator $e^{-i \mathcal{H} \tau}$ to $\ket{u(0)}$ via matrix-vector calculation. 
Dotted lines represents the solutions calculated by the FDM.
Since the velocity is chosen to have a positive value $v_1=1$, we observe that the scalar field $u(t, x_i)$ moves toward the positive direction as time passes. 
Also, the solid lines in Fig.~\ref{fig:advect_1d}(b) illustrate the absolute errors between the amplitudes of the quantum state $\ket{u(t)}$ by the proposed circuits and those calculated by FDM; moreover, the dashed lines illustrate the absolute errors between the amplitudes calculated via applying the matrix exponential operator $e^{-i \mathcal{H} t}$ to $\ket{u(0)}$ via matrix-vector calculation and those calculated by FDM.
The figures imply that the solutions obtained by the three approaches agree well, which demonstrated that the advection equation discretized by the FDM can be simulated by the Hamitlonian simulation algorithm on quantum circuits under acceptable Trotter errors. 
The oscillation occurred in some spatial region where $u(t, x_i)$ had the sharp gradients, is due to the central differencing scheme, and it is well-known that it can be resolved by using the upwind differencing scheme~\cite{allaire2007numerical}.
However, such differencing scheme does not retain the Hermitian property of the Hamiltonian as we mentioned in Section~\ref{sec:advect_sch}, that is, $\mathcal{H}^\dagger \neq \mathcal{H}$ for $\mathcal{H} = -i \bm{v} \cdot \nabla$ when we use $D^{+}$ or $D^{-}$ in Eqs.~\eqref{eq:forward diff operator} and \eqref{eq:backward diff operator} to discretize $\nabla$ because of the relationship of $(D^{+})^\dagger = D^{-}$.
We will address this issue in the future research.

\begin{figure*}[t]
    \centering
    \includegraphics[width=0.8\textwidth]{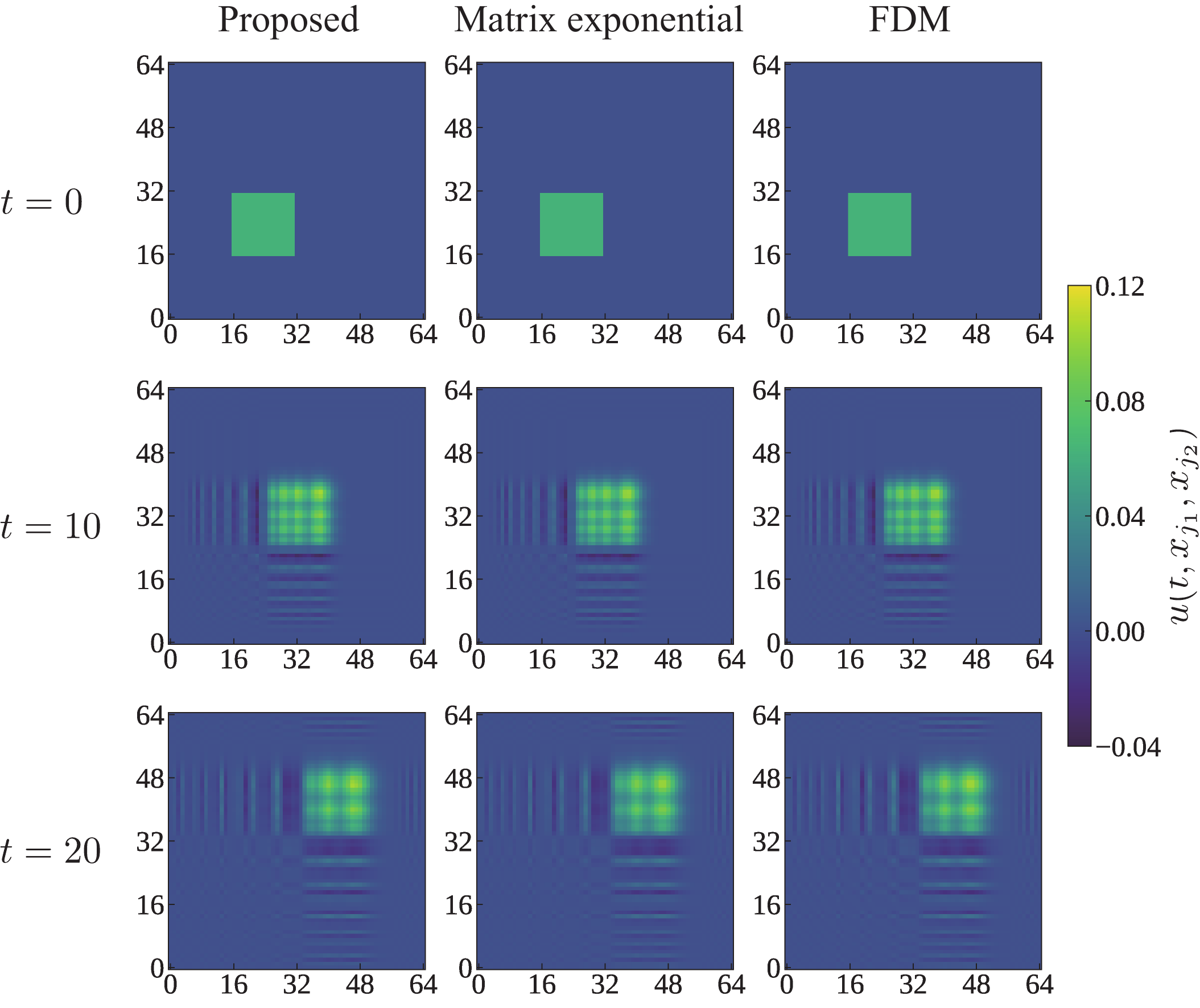}
    \caption{Simulation results of the advection equation with periodic BC in two dimensions. The color plot represents the values of the scalar field $\ket{u(t)}$ calculated by the three approaches.}
    \label{fig:advect_2d}
\end{figure*}

Next, Fig.~\ref{fig:advect_2d} shows the simulation results of the advection equation with periodic BC in two dimensions.
We set $n=6$, $\tau=0.1$, $T=20$, $l=1$, and $\bm{v}=[v_1, v_2]=[1,1]$
to setup the problem and used \texttt{statevector\_simulator}, to run the quantum circuits.
The number of qubits is $2n=12$.
We chose the time increment parameter to be $0.01$ for the FDM.
As an initial state, we prepare
\begin{align}
    &\ket{u(0)} \nonumber \\
    &= \frac{1}{2^{n-2}} \sum_{j_1=0}^{2^{n-2}-1} \sum_{j_2=0}^{2^{n-2}-1} \ket{0} \otimes \ket{1} \otimes  \ket{j_1} \otimes \ket{0} \otimes \ket{1} \otimes \ket{j_2}.
\end{align}
Figure~\ref{fig:advect_2d} also include this initial state as the amplitude at $t=0$ to visualize the dynamics from the initial state.
The left column in Fig.~\ref{fig:advect_2d} represents the amplitudes of the quantum state $\ket{u(t)}$ generated by the proposed circuits described in Section~\ref{sec:advect_sch}.
The center column represents the solutions directly obtained by applying the matrix exponential operator $e^{-i \mathcal{H} \tau}$ to $\ket{u(0)}$ by matrix-vector calculation. 
The right column represents solutions calculated by the FDM.
Each row represents the solutions at $t=0$, $t=10$ and $t=20$, respectively.
We observe that the distribution of the scalar field $\ket{u}$ moves toward upper right direction according to the setting of $\bm{v}=[v_1, v_2]=[1,1]$.
Although several numerical oscillation occurred due to the central differencing scheme, we again find that the solutions obtained by the three approaches agree well, which demonstrate the validity of the proposed method.
Since the solution in the center column is obtained by directly applying the matrix exponential operator, it does not include the error in the numerical time integration.
The difference between the left and center columns comes from the Trotter error, while the difference between the center and right columns comes from the time integral by the forward Euler scheme for the FDM.
These errors can be decreased by using a smaller time increment parameter.
That is, the time evolution operator $\exp(-i \mathcal{H} \tau)$ is directly applied to the state $\ket{u(t)}$ in the center column, while the forward Euler scheme~\cite{allaire2007numerical} is used to proceed time in the right column.

\subsection{Wave equation}

Next, we show Hamiltonian simulation for solving the wave equation with the mixed BC in one dimension and with the periodic BC in two dimensions.
We here again compare the simulation results of our proposed method with those obtained by directly using the matrix exponential operator $e^{-i \mathcal{H} \tau}$ to $\ket{\psi(0)}$ and those by the fully classical finite difference method (FDM) with the central differencing scheme.

\begin{figure*}[t]
    \centering
    \includegraphics[width=0.8\textwidth]{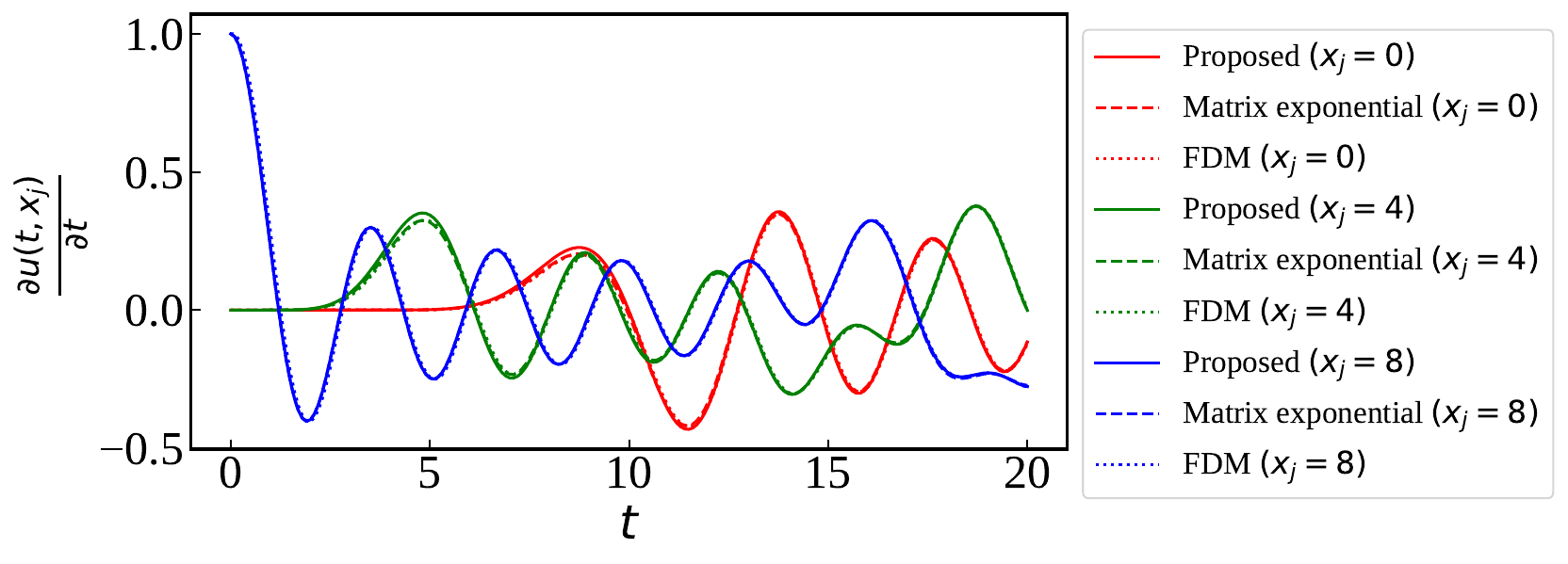}
    \caption{Simulation results of the wave equation with the mixed BC in one dimension. Solid and dashed lines represent the amplitudes of the quantum state $\ket{\psi(t)}$ prepared by the proposed circuits and that obtained by direct applying the matrix exponential operator $e^{-i \mathcal{H} \tau}$ to $\ket{\psi(0)}$, respectively. Dotted lines represent the solutions calculated by the finite difference method (FDM).}
    \label{fig:wave_1d}
\end{figure*}
\begin{figure*}[t]
    \centering
    \includegraphics[width=0.8\textwidth]{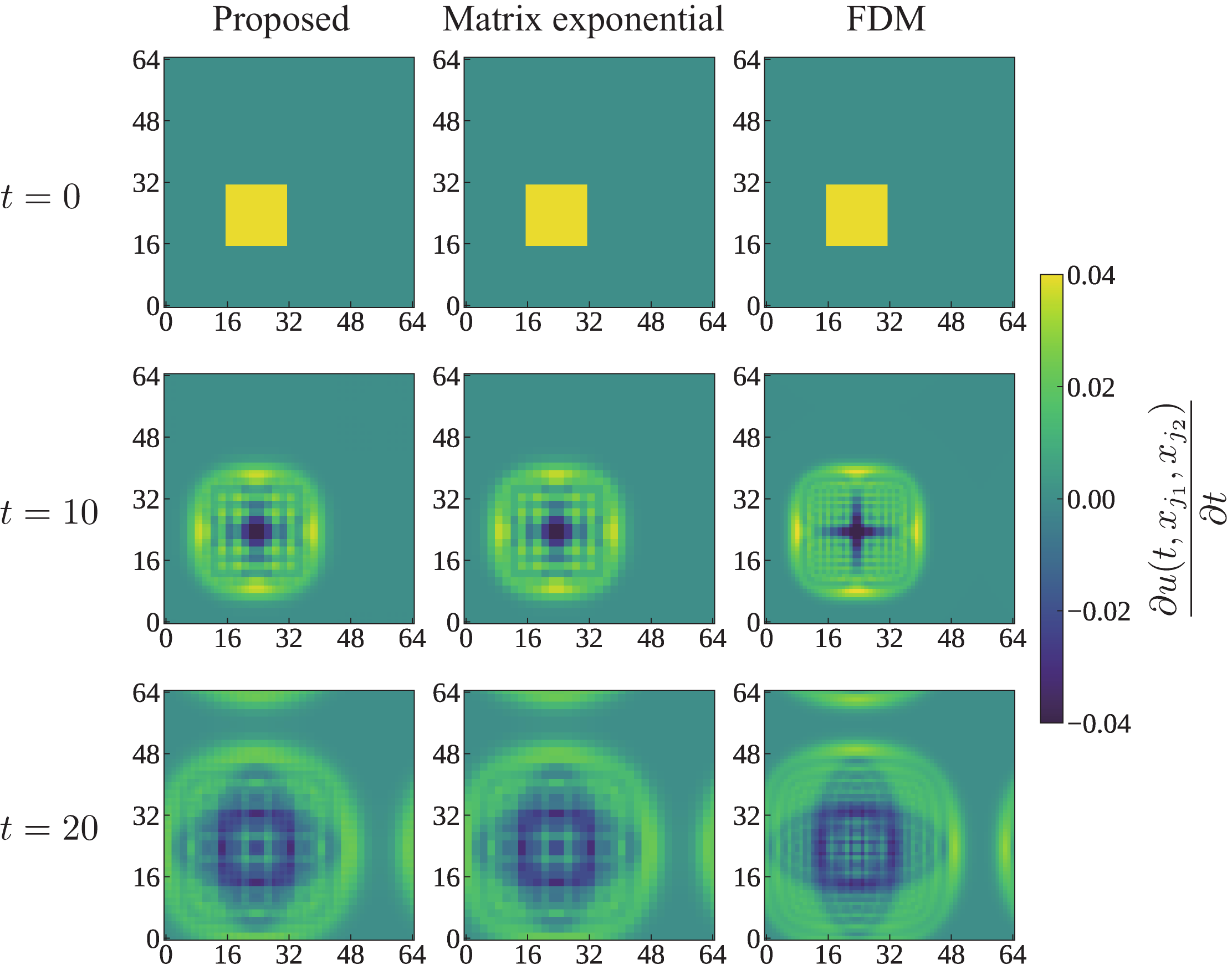}
    \caption{Simulation results of the wave equation with the periodic BC in two dimensions. Each plot in the column represents the scalar field $\partial u / \partial t$ calculated by each approach.}
    \label{fig:wave_2d}
\end{figure*}

Figure~\ref{fig:wave_1d} illustrates the simulation results of the wave equation with the mixed BC in one dimension, specifically the Dirichlet BC for the left side and the Neumann BC for the right side.
We set $n=4$, $\tau=0.1$, $T=20$, $l=1$ and $c=1$ to setup the problem and used \texttt{statevector\_simulator} to run quantum circuits.
The number of qubits for encoding this problem is $n+1=5$. 
We use a time increment parameter $0.1$ for the FDM.
As an initial state, we prepare
\begin{align}
    \ket{\psi(0)} = \ket{0} \otimes \ket{1} \otimes \ket{0}^{\otimes (n-1)},
\end{align}
which corresponds to the initial condition
\begin{align}
    &u(0, x_j) = 0 \\
    &\pdif{u(0, x_j)}{t} = 
    \begin{cases}
        1 & \text{ for } j = 2^{n-1} \\
        0 & \text{ otherwise}.
    \end{cases}
\end{align}
Solid, dashed, and dotted lines represent components of solutions corresponding to $\partial u(t, x_j) / \partial t$ prepared by the proposed circuits described in Section~\ref{sec:wave_sch}, the direct multiplication of the matrix exponential operator by $\ket{\psi(0)}$, and the FDM, respectively.
Figure~\ref{fig:wave_1d} clearly illustrates that the solutions obtained by the three approaches agree well, demonstrating that the wave equation could be simulated as the Schr\"{o}dinger equation as well and actually be implemented on quantum circuits under acceptable Trotter errors.

\begin{figure*}[t]
    \centering
    \includegraphics[width=0.8\textwidth]{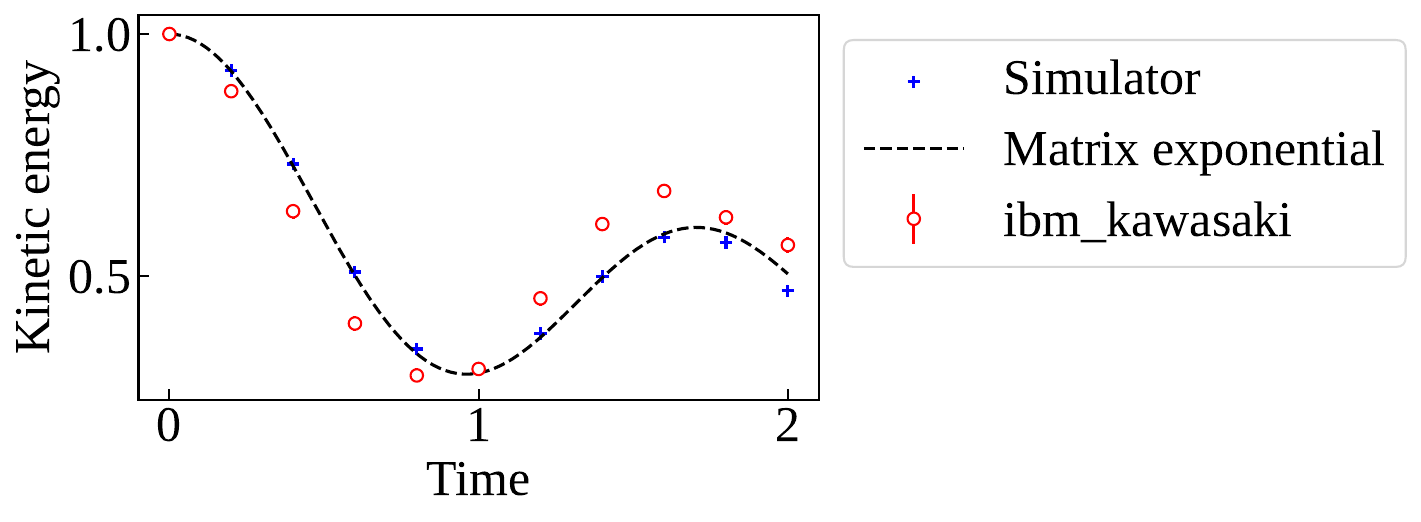}
    \caption{Experimental results of the wave equation with the mixed BC in one dimension by \texttt{ibm\_kawasaki}. 
    Red circle represents the expectation value estimated by 4000 shots along with the error bar representing the 95\% confidence interval of the expectation value under the shot noise. Note that Error bars are small enough to be visible.
    Blue $+$ represents the result of \texttt{statevector\_simulator}.
    Black dashed line represents results obtained by the direct multiplication of the matrix exponential operator $e^{-i \mathcal{H} \tau}$ by $\ket{\psi(0)}$.}
    \label{fig:wave_1d_realdevice}
\end{figure*}

Figure~\ref{fig:wave_2d} illustrates the simulation results of the wave equation with the periodic BC in two dimensions.
We set $n=6$, $\tau=0.1$, $T=20$, $l=1$ and $c=1$ to setup the problem and used \texttt{statevector\_simulator} to run quantum circuits.
The number of qubits to implement this problem is $2n+1=13$.
Here, we used the central difference operator $D^{\pm}$ to define the Hamiltonian.
We use a time increment parameter $0.1$ for the FDM.
As an initial state, we prepare
\begin{align}
    &\ket{\psi(0)} \nonumber \\
    &= \ket{0} \otimes \frac{1}{2^{n-2}} \sum_{j_1=0}^{2^{n-2}-1} \sum_{j_2=0}^{2^{n-2}-1} \ket{0} \ket{1} \ket{j_1} \ket{0} \ket{1} \ket{j_2},
\end{align}
which reflectes the initial condition of 
\begin{align}
    &u(t, x_{j_1}, x_{j_2}) = 0 \\
    &\pdif{u(t, x_{j_1}, x_{j_2})}{t} = \begin{cases}
        1 & \text{ for } 2^{n-2} \leq x_{j_1}, x_{j_2} < 2^{n-1}\\
        0 & \text{ otherwise}.
    \end{cases}
\end{align}
Figure~\ref{fig:wave_2d} also include this initial state as the amplitude at $t=0$ to visualize the dynamics from the initial state.
The left and center columns in Fig.~\ref{fig:wave_2d} represent components of solutions corresponding to $\partial u(t, x_{j_1}, x_{j_2}) / \partial t$ obtained by the proposed quantum circuit and the direct calculation of the matrix exponential operator, respectively.
The right column represents solutions calculated by the FDM.
Each row represents the solutions at $t=0$, $t=10$ and $t=20$, respectively.
The results of left and center columns clearly coincide, implying that the proposed quantum circuit could accurately simulate the Schr\"{o}dinger equation derived from the wave equation.
However, these result slightly differ from that by the FDM; the results by the proposed method look rougher than those by the FDM.
This comes from the difference of finite difference schemes.
For the fully classical FDM, we discretized the Laplacian $\nabla^2$ of the wave equation by $D^{\Delta}$ while for the Schr\"{o}dinger equation derived from the wave equation, we discretized the gradient operator $\nabla$ by $D^{\pm}$, which corresponds to discretizing the Laplacian $\nabla^2$ by $(D^{\pm})^2$.
Thus, the discretized equations do not exactly match and accordingly have errors that can be reduced by decreasing the spatial interval $l$.
Although there exists such numerical error, overall, the results well grasp the property of wave propagation.

Finally, we provide an experimental result conducted on a real quantum device.
We here again study the wave equation with the mixed BC in one dimension setting $n=2$, $\tau=0.2$, $T=2$, $l=1$ and $c=1$. 
The number of qubits for encoding this problem is $n+1=3$.
Although the number of qubits used here is quite small due to the limitation of noise resilience of current hardware, we show the experimental result to demonstrate that our quantum algorithm is implementable on a current device owing to the explicit construction of quantum circuits by elementary gates of the hardware.
We compare the simulation results obtained from a real device, \texttt{statevector\_simulator}, and the direct multiplication of the matrix exponential operator $e^{-i \mathcal{H} \tau}$ by $\ket{\psi(0)}$.
Although we illustrated all probability amplitudes of quantum states prepared by each approach in the preceding results for validation, accessing all amplitudes is unrealistic for practical use because of the need of the quantum state tomography as we mentioned in Section~\ref{sec:complexity}.
Here, we evaluated the expectation value of the following specific observable:
\begin{align}
    \mathcal{O} = \frac{1}{2} (Z + I) \otimes I^{\otimes n},
\end{align}
which is the kinetic energy of the system.
As an initial state, we simply used
\begin{align}
    \ket{\psi(0)} = \ket{0}^{\otimes n} \otimes \ket{1},
\end{align}
which corresponds to the initial condition
\begin{align}
    &u(0, x_j) = 0 \\
    &\pdif{u(0, x_j)}{t} = 
    \begin{cases}
        1 & \text{ for } j = 1 \\
        0 & \text{ otherwise}.
    \end{cases}
\end{align}

Figure~\ref{fig:wave_1d_realdevice} illustrates the experimental results.
We used 4000 shots to estimate the expectation $\braket{\psi(t) | \mathcal{O} | \psi(t)}$ at each time.
We used the dynamical decoupling (DD) techniques~\cite{Ezzell2023, pokharel2018demonstration, jurcevic2021demonstration} with super-Hahn echo~\cite{Hahn1950, Ezzell2023} to suppress the decoherence, and also used the readout error mitigation technique, specifically Twirled Readout Error eXtinction (TREX)~\cite{van2022model}.
The number of non-local gates in the quantum circuit was 120 at the maximum simulation time $t=2$ for this problem.
Note that error bars representing the 95\% confidence interval of the expectation value under the shot noise are small enough to be visible since the number of shots, 4000, is large enough to estimate the observable of the one-local Pauli operator.
That is, the difference between solutions obtained by the real device (red circles) and the simulator (blue $+$ symbols) is attributed to the hardware noise and the error mitigation.
Nonetheless, Fig.~\ref{fig:wave_1d_realdevice} clearly illustrates that the solutions obtained by the real device, \texttt{ibm\_kawasaki}, captures the result of simulator, which demonstrated that proposed method actually be implementable on a real device.

\section{Conclusions} \label{sec:conclusion}

This paper proposed a scalable quantum circuit implementation method of Hamiltonian simulation for partial differential equations.
The key technique for efficiently implementing the time evolution operators on a circuit is to diagonalize each term of Hamiltonian using the Bell basis.
We provided concrete quantum circuit representation for Hamiltonian simulation driven by differential operators along with the space and time complexity estimation.
For demonstrating the validity of the proposed method, we focused on two partial differential equations, namely the advection and wave equations, which were converted into the Schr\"{o}dinger equation to be simulated by Hamiltonian simulation. 
In numerical experiments, we confirmed that the solutions obtained by the proposed method agreed well with the solution calculated by the fully classical finite difference method. 
This means that the advection and wave equations could be simulated as the Schr\"{o}dinger equation and actually be implemented on quantum circuits under acceptable Trotter errors.

In the present study, we focused on the advection and wave equations, because they are conservative under appropriate BCs and thus can be exactly converted to the Schr\"{o}dinger equation.
However, because most practical systems are non-conservative, it is required to extend the present approach to systems which cannot directly be described as the Schr\"{o}dinger equation.
Possible approaches are, for example, to describe the target systems as an imaginary time evolution of the Schr\"{o}dinger equation~\cite{kosugi2022imaginary, leadbeater2023non} or as the open quantum system governed by the Lindblad equation~\cite{schlimgen2022quantum}, both of which realize the non-unitary time evolution.
We would like to conduct our future research toward such direction.
Also, we would like to investigate the relationship among BCs, finite difference schemes, and the resulting Hermitian property in our future research.

\section*{Acknowledgment}
This work is supported by MEXT Quantum Leap Flagship Program Grant Number JPMXS0118067285 and JP- MXS0120319794, and JSPS KAKENHI Grant Number 20H05966. 
We acknowledge the use of IBM Quantum services for this work. The views expressed are those of the authors, and do not reflect the official policy or position of IBM or the IBM Quantum team.
We would like to thank Takashi Imamichi and Jun Doi of IBM Quantum, IBM Research -- Tokyo for the fruitful discussions on theoretical analysis and implementation.

\onecolumngrid
\appendix
\section{Finite difference operators of higher-order approximations} \label{sec:higher-order-FDM}
Using shift operators $S^{-}$ and $S^{+}$ in Eqs.~\eqref{eq:shift_m} and \eqref{eq:shift_p}, respectively, we can construct finite difference operators of higher-order approximations.
For instance, the second-order forward difference operator for the first derivative, denoted by $D^{+}_2$ acting on the nodal values $\bm{u} \in \mathbb{R}^{2^n}$ as 
\begin{align}
    \left( D^{+}_2 \bm{u} \right)_j = \frac{-3u_j + 4u_{j+1} - u_{j+2}}{2l},
\end{align}
can be represented as
\begin{align}
    D^{+}_2 = \frac{1}{2l} \left(-3 I^{\otimes n} + 4 S^{-} - (S^{-})^2 \right),
\end{align}
which is a qubit operator of an $n$-qubit system.
Actually, the following relationship holds:
\begin{align}
    D^{+}_2 \ket{u} &= \frac{1}{2l} \left(-3 I^{\otimes n} + 4 S^{-} - (S^{-})^2 \right) \sum_{j=0}^{2^n-1} u_j \ket{j} \nonumber \\
    &= \sum_{j=0}^{2^n-1} \frac{-3u_j + 4u_{j+1} - u_{j+2}}{2l} \ket{j},
\end{align}
where $u_{2^n} = u_{2^n+1} = 0$, which exactly corresponds to the second-order forward difference operator for the first derivative with the Dirichlet boundary condition.
Similarly, the qubit operator for the second-order backward difference operator for the first derivative, denoted by $D^{-}_2$, is represented as
\begin{align}
    D^{-}_2 = \frac{1}{2l} \left(3 I^{\otimes n} - 4 S^{+} + (S^{+})^2 \right).
\end{align}
Other finite difference operators for the higher-order derivatives can also be constructed using the identity operator and the shift operators $S^{-}$ and $S^{+}$.

\section{Self-adjointness of the operator $\mathcal{H}$} \label{sec:self-adjoint}
\subsection{Operator of the Advection equation} \label{sec:self-adjoint_advect}
Let $u$ and $\tilde{u}$ be scalar fields in the Hilbert space $\mathcal{U}$ defined on a domain $\Omega$.
Applying the integration by parts, we obtain the following relationship about the inner product in the space $\mathcal{U}$: 
\begin{align} \label{eq:inner_product_advect}
    \left\langle \tilde{u}, \left(-i \bm{v} \cdot \nabla \right) u \right\rangle 
    &= \int_\Omega \tilde{u}^\ast \left(-i \bm{v} \cdot \nabla \right) u \mathrm{d} \bm{x} \nonumber \\
    &= \int_{\partial \Omega} \left( -i \bm{n} \cdot \bm{v} \tilde{u}^\ast u \right) \mathrm{d} \bm{x} + \int_\Omega \left( i \bm{v} \cdot \nabla \right) \tilde{u}^\ast u \mathrm{d} \bm{x} \nonumber \\
    &= \int_{\partial \Omega} \left( -i \bm{n} \cdot \bm{v} \tilde{u}^\ast u \right) \mathrm{d} \bm{x} + \left\langle \left(-i \bm{v} \cdot \nabla \right)\tilde{u},  u \right\rangle,
\end{align}
where $\ast$ represents the complex conjugate and $\bm{n}$ is the normal vector pointing outward at the boundary.
If the first term in the last line is zero, the operator $-i \bm{v} \cdot \nabla$ is self-adjoint.
This is the case when the domain $\Omega$ is a unit cell and the periodic BC is imposed on the boundaries that satisfy $\bm{n} \cdot \bm{v} \neq 0$, i.e.,
\begin{align}
    \mathcal{U} = \{ u \in H^1 (\Omega) ~|~ u \text{ is periodic on } \partial \Omega \text{ s.t. } \bm{n} \cdot \bm{v} \neq 0 \},
\end{align}
where $H^1(\Omega)$ is the Sobolev space defined on the domain $\Omega$.
The first term of the last line in Eq.~\eqref{eq:inner_product_advect} also vanishes under the Dirichlet BC, i.e., 
\begin{align}
    \mathcal{U} = \{ u \in H^1 (\Omega) ~|~ u = 0 \text{ on } \partial \Omega \}.
\end{align}
For such $\mathcal{U}$, however, only trivial solution $u = 0 \in \mathcal{U}$ satisfies the advection equation.

\subsection{Operator of the wave equation} \label{sec:self-adjoint_wave}
Let $\bm{\psi}(t, \bm{x})$ and $\tilde{\bm{\psi}}(t, \bm{x})$ be vector fields defined in Sec.~\ref{sec:wave_sch}.
Applying the integration by parts, we obtain the following relationship about the inner product:
\begin{align} \label{eq:inner_product_wave}
    \left\langle \tilde{\bm{\psi}}, \mathcal{H} \bm{\psi} \right\rangle 
    &= \int_0^T \int_\Omega \tilde{\bm{\psi}}^\dagger \mathcal{H} \bm{\psi} \mathrm{d} \bm{x} \mathrm{d} t \nonumber \\
    &= \int_0^T \int_\Omega ic^2 \left( \pdif{\tilde{u}}{t}^\ast \nabla^2 u + \nabla \tilde{u}^\ast \cdot \nabla \pdif{u}{t} \right) \mathrm{d} \bm{x} \mathrm{d} t \nonumber \\
    &= \int_0^T \int_{\partial \Omega} ic^2 \left( \pdif{\tilde{u}}{t}^\ast \bm{n} \cdot \nabla u + \bm{n} \cdot \nabla \tilde{u}^\ast \pdif{u}{t} \right) \mathrm{d} \bm{x} \mathrm{d} t - \int_0^T \int_\Omega ic^2 \left( \nabla \pdif{\tilde{u}}{t}^\ast \cdot \nabla u + \nabla^2 \tilde{u}^\ast \pdif{u}{t} \right) \mathrm{d} \bm{x} \mathrm{d} t \nonumber \\
    &= \int_0^T \int_{\partial \Omega} ic^2 \left( \pdif{\tilde{u}}{t}^\ast \bm{n} \cdot \nabla u + \bm{n} \cdot \nabla \tilde{u}^\ast \pdif{u}{t} \right) \mathrm{d} \bm{x} \mathrm{d} t + \int_0^T \int_\Omega \left( \mathcal{H} \tilde{\bm{\psi}} \right)^\dagger \bm{\psi} \mathrm{d} \bm{x} \mathrm{d} t \nonumber \\
    &= \int_0^T \int_{\partial \Omega} ic^2 \left( \pdif{\tilde{u}}{t}^\ast \bm{n} \cdot \nabla u + \bm{n} \cdot \nabla \tilde{u}^\ast \pdif{u}{t} \right) \mathrm{d} \bm{x} \mathrm{d} t + \left\langle \mathcal{H} \tilde{\bm{\psi}}, \bm{\psi} \right\rangle,
\end{align}
where $\ast$ is the complex conjugate and $\bm{n}$ is the normal vector pointing outward at the boundary.
If the first term in the last line is zero, the operator $\mathcal{H}$ in Sec.~\ref{sec:wave_sch} for the wave equation is self-adjoint.
This term vanishes when the domain $\Omega$ is a unit cell and the periodic BC is imposed on the boundaries, i.e., $u \in \mathcal{U}$ where
\begin{align}
    \mathcal{U} = \{ u \in H^2 (\Omega) ~|~ u \text{ and } \bm{n} \cdot \nabla u \text{ is periodic on } \partial \Omega \},
\end{align}
where $H^2(\Omega)$ is the Sobolev space defined on the domain $\Omega$.
The condition for the self-adjointness of the operator $\mathcal{H}$ is also satisfied under the mixed BCs of the Dirichlet and Neumann BCs for $u$, i.e., when $u \in \mathcal{U}$ where
\begin{align}
    \mathcal{U} = \{ u \in H^2 (\Omega) ~|~ u = 0 \text{ on } \Gamma_\mathrm{D} \text{ and } \bm{n} \cdot \nabla u = 0 \text{ on } \Gamma_\mathrm{N} \},
\end{align}
with $\Gamma_\mathrm{D} \cup \Gamma_\mathrm{N} = \partial \Omega$, $\Gamma_\mathrm{D} \neq \emptyset$ and $\Gamma_\mathrm{N} \neq \emptyset$.

\section{Proofs of Lemmas~\ref{lem:diff_circ} and \ref{lem:diff_circ_d-dim}} \label{sec:proof_lemma}
Here, we provide details of proofs of Lemmas~\ref{lem:diff_circ} and \ref{lem:diff_circ_d-dim}.
For convenience, we first recall Lemma~\ref{lem:diff_circ}.
\setcounter{lem}{0}
\begin{lem} 
    Consider the Schr\"{o}dinger equation 
    $d \ket{u(t)}/dt = -i \mathcal{H} \ket{u(t)}$ such that
    the Hamiltonian $\mathcal{H}$ is given by Eq.~\eqref{eq:hamiltonian_diffop}.
    The time evolution operator $\exp(-i \mathcal{H} \tau)$ with the time increment $\tau$ can be approximated by the unitary $V$ in Eq.~\eqref{eq:trotter_general}, and its explicit circuit implementation is shown in Fig.~\ref{fig:circuit}. 
    Moreover, the approximation error in the sense of the operator norm is upper bounded as 
    \begin{align}
        \left\Vert \exp(-i \mathcal{H} \tau) - V( \gamma \tau, \lambda) \right\Vert &\leq \frac{\gamma^2 \tau^2 (n-1)}{2}.
    \end{align}
\end{lem}
\begin{proof}
    As mentioned in the main text, our task is to evaluate the commutators of the Hamiltonian.
    For $n>j>j'>1$, the terms of the shift operators $e^{i\lambda} s^{-}_j + e^{-i\lambda}s^{+}_j$ and $e^{i\lambda} s^{-}_{j'} + e^{-i\lambda} s^{+}_{j'}$ commute as
    \begin{align}
        \left[ e^{i\lambda} s^{-}_j + e^{-i\lambda} s^{+}_j, e^{i\lambda} s^{-}_{j'} + e^{-i\lambda} s^{+}_{j'} \right] & = \left[ e^{i\lambda} I^{\otimes (n-j)} \otimes \sigma_{01} \otimes \sigma_{10}^{\otimes (j-1)} + e^{-i\lambda} I^{\otimes (n-j)} \otimes \sigma_{10} \otimes \sigma_{01}^{\otimes (j-1)}, \right. \nonumber \\
        & \left. \quad e^{i\lambda} I^{\otimes (n-j')} \otimes \sigma_{01} \otimes \sigma_{10}^{\otimes (j'-1)} + e^{-i\lambda} I^{\otimes (n-j')} \otimes \sigma_{10} \otimes \sigma_{01}^{\otimes (j'-1)} \right] \nonumber \\
        &= 0.
    \end{align}
    For $j>j'=1$, the commutator is rearranged as
    \begin{align} \label{eq:commutator_1}
        & \left[ e^{i\lambda} s^{-}_j + e^{-i\lambda} s^{+}_j, e^{i\lambda} s^{-}_1 + e^{-i\lambda} s^{+}_1 \right] \nonumber \\
        & = \left[ e^{i\lambda} I^{\otimes (n-j)} \otimes \sigma_{01} \otimes \sigma_{10}^{\otimes (j-1)} + e^{-i\lambda} I^{\otimes (n-j)} \otimes \sigma_{10} \otimes \sigma_{01}^{\otimes (j-1)}, e^{i\lambda} I^{\otimes (n-1)} \otimes \sigma_{01} + e^{-i\lambda} I^{\otimes (n-1)} \otimes \sigma_{10} \right] \nonumber \\
        &= - I^{\otimes (n-j)} \otimes \left( e^{2i\lambda} \sigma_{01} \otimes \sigma_{10}^{\otimes (j-2)} - e^{-2i\lambda} \sigma_{10} \otimes \sigma_{01}^{\otimes (j-2)} \right) \otimes Z \nonumber \\
        &= i I^{\otimes (n-j)} \otimes \left( \frac{ \ket{0} \! \ket{1}^{\otimes (j-2)} + e^{-i(2\lambda + \pi/2)} \ket{1} \! \ket{0}^{\otimes (j-2)} }{\sqrt{2}} \frac{ \bra{0} \! \bra{1}^{\otimes (j-2)} + e^{i(2\lambda + \pi/2)} \bra{1} \! \bra{0}^{\otimes (j-2)} }{\sqrt{2}} \right. \nonumber \\
        & \qquad \qquad \qquad \qquad \left. - \frac{ \ket{0} \! \ket{1}^{\otimes (j-2)} - e^{-i(2\lambda + \pi/2)} \ket{1} \! \ket{0}^{\otimes (j-2)} }{\sqrt{2}} \frac{ \bra{0} \! \bra{1}^{\otimes (j-2)} - e^{i(2\lambda + \pi/2)} \bra{1} \! \bra{0}^{\otimes (j-2)} }{\sqrt{2}} \right) \otimes Z \nonumber \\
        &= iI^{\otimes (n-j)} \otimes \left( U_{j-1}(-2\lambda-\pi/2) \left( Z \otimes \ket{1} \! \bra{1}^{\otimes (j-2)} \right) U_{j-1}(-2\lambda-\pi/2)^\dagger \right) \otimes Z.
    \end{align}
    Thus, the terms of the Hamiltonian, $\gamma(e^{i\lambda} s^{-}_j + e^{-i\lambda} s^{+}_j)$, can be grouped into those for $j>1$ and $j=1$.
    Since the unitary does not change the operator norm, the Trotter error is upper bounded using the result of Ref.~\cite[Proposition 9]{childs2021theory}, as follows: %
    \begin{align}
        \left\Vert \exp(-i \mathcal{H} \tau) - V( \gamma \tau, \lambda ) \right\Vert
        &\leq \frac{\gamma^2 \tau^2}{2} \left\Vert \left[ \sum_{j=2}^n \left( e^{i\lambda} s^{-}_j + e^{-i\lambda} s^{+}_j \right), e^{i\lambda} s^{-}_1 + e^{-i\lambda} s^{+}_1 \right] \right\Vert \nonumber \\
        &\leq \frac{\gamma^2 \tau^2}{2} \sum_{j=2}^n \left\Vert \left[  e^{i\lambda} s^{-}_j + e^{-i\lambda} s^{+}_j, e^{i\lambda} s^{-}_1 + e^{-i\lambda} s^{+}_1 \right] \right\Vert \nonumber \\
        &= \frac{\gamma^2 \tau^2 (n-1)}{2}
    \end{align}
\end{proof}

Next, we recall Lemma~\ref{lem:diff_circ_d-dim}.
\setcounter{lem}{3}
\begin{lem}
    Let the Hamiltonian $\mathcal{H}$ consist of finite difference operators for a $dn$-qubit system as
    \begin{align} %
        \mathcal{H} = \gamma \sum_{\alpha=1}^d \sum_{j=1}^n \eta_\alpha \left( e^{i\lambda_\alpha} (s^{-}_j)_\alpha + e^{-i\lambda_\alpha} (s^{+}_j)_\alpha  \right),
    \end{align}
    where $\gamma \in \mathbb{R}$ is a scale parameter, $\lambda_\alpha \in \mathbb{R}$ is the phase parameter and 
    \begin{align}
        (s^\mu_j)_\alpha = I^{\otimes (\alpha-1)n} \otimes s^\mu_j \otimes I^{\otimes (d-\alpha)n},
    \end{align}
    for $\mu \in \{ -, + \}$.
    The time evolution operator $\exp(-i \mathcal{H} \tau)$ with the time increment $\tau$ can be approximated by the unitary $\bigotimes_{\alpha=1}^d V(\gamma \eta_\alpha \tau, \lambda_\alpha)$. 
    The approximation error is upper bounded in the sense of the operator norm as
    \begin{align}
        \left\Vert \exp \left(-i \mathcal{H} \tau \right) - \bigotimes_{\alpha=1}^d V(\gamma \eta_\alpha \tau, \lambda_\alpha) \right\Vert \leq \frac{\gamma^2 \tau^2 (n-1) \sum_{\alpha=1}^d \eta_\alpha^2}{2}.
    \end{align}
\end{lem}
\begin{proof}
    From Eq.~\eqref{eq:shift_general}, we obtain
    \begin{align}
        \exp \left(-i \mathcal{H} \tau \right) &= \exp \left( -i \gamma \tau \sum_{\alpha=1}^d \sum_{j=1}^n \eta_\alpha (e^{i\lambda_\alpha} (s^{-}_j)_\alpha + e^{-i\lambda_\alpha} (s^{+}_j)_\alpha) \right) \nonumber \\
        &\approx \prod_{\alpha=1}^d \prod_{j=1}^{n} \exp \left( -i \gamma \eta_\alpha \tau I^{\otimes (\alpha - 1)n } \otimes I^{ \otimes (n-j)} \otimes \left( U_j (-\lambda_\alpha) \left( Z \otimes \ket{1} \! \bra{1}^{\otimes (j-1)}\right) U_j (-\lambda_\alpha)^\dagger \right) \otimes I^{\otimes (d - \alpha)n} \right) \nonumber \\
        & = \prod_{\alpha=1}^d \prod_{j=1}^{n} I^{\otimes (\alpha - 1)n } \otimes I^{ \otimes (n-j)} \otimes \left( U_j (-\lambda_\alpha) \crz^{1, \dots, j-1}_j \left( 2\gamma \tau \eta_\alpha \right) U_j (-\lambda_\alpha)^\dagger \right) \otimes I^{\otimes (d - \alpha)n} \nonumber \\
        & = \prod_{\alpha=1}^d I^{\otimes (\alpha - 1)n } \otimes V(\gamma \eta_\alpha \tau, \lambda_\alpha) \otimes I^{\otimes (d - \alpha)n} \nonumber \\
        &= \bigotimes_{\alpha=1}^d V\left( \gamma \tau \eta_\alpha, \lambda_\alpha \right).
    \end{align}
    The approximation error of the Lie-Trotter-Suzuki decomposition is upper bounded by the operator norm of commutators of Hamiltonian~\cite[Proposition 9]{childs2021theory}.
    Since the terms of the Hamiltonian $e^{i\lambda_\alpha} (s^{-}_j)_\alpha + e^{-i\lambda_\alpha} (s^{+}_j)_\alpha$ and $e^{i\lambda_{\alpha'}} (s^{-}_j)_{\alpha'} + e^{-i\lambda_{\alpha'}} (s^{+}_j)_{\alpha'}$ commute for $\alpha \neq \alpha'$, we obtain from the discussion in the proof of Lemma~\ref{lem:diff_circ},
    \begin{align}
        &\left\Vert \exp \left(-i \mathcal{H} \tau \right) - \bigotimes_{\alpha=1}^d V\left( \gamma \tau \eta_\alpha, \lambda_\alpha \right) \right\Vert \nonumber \\
        &\leq \frac{\gamma^2 \tau^2}{2} \left\Vert \left[ \sum_{\alpha=1}^d \sum_{j=2}^n \eta_\alpha \left( e^{i\lambda_\alpha} (s^{-}_j)_\alpha + e^{-i\lambda_\alpha} (s^{+}_j)_\alpha \right), \sum_{\alpha'=1}^d \eta_{\alpha'} \left( e^{i\lambda_{\alpha'}} (s^{-}_1)_{\alpha'} + e^{-i\lambda_{\alpha'}} (s^{+}_1)_{\alpha'} \right) \right] \right\Vert \nonumber \\
        &\leq \frac{\gamma^2 \tau^2}{2} \sum_{\alpha=1}^d \sum_{\alpha'=1}^d \sum_{j=2}^n \eta_\alpha \eta_{\alpha'} \left\Vert \left[ \left( e^{i\lambda_\alpha} (s^{-}_j)_\alpha + e^{-i\lambda_\alpha} (s^{+}_j)_\alpha \right), e^{i\lambda_{\alpha'}} (s^{-}_1)_{\alpha'} + e^{-i\lambda_{\alpha'}} (s^{+}_1)_{\alpha'} \right] \right\Vert \nonumber \\
        &= \frac{\gamma^2 \tau^2}{2} \sum_{\alpha=1}^d \sum_{j=2}^n \eta_\alpha^2 \left\Vert \left[ e^{i\lambda_\alpha} (s^{-}_j)_\alpha + e^{-i\lambda_\alpha} (s^{+}_j)_\alpha, e^{i\lambda_{\alpha}} (s^{-}_1)_{\alpha} + e^{-i\lambda_{\alpha}} (s^{+}_1)_{\alpha} \right] \right\Vert \nonumber \\
        &= \frac{\gamma^2 \tau^2 (n-1) \sum_{\alpha=1}^d \eta_\alpha^2}{2}.
    \end{align}
\end{proof}

\section{Proof of Lemma~\ref{lem:diff_circ_second}} \label{sec:proof_lemma_second}

Here, we provide the details of proofs of Lemma~\ref{lem:diff_circ_second}.
We first recall Lemma~\ref{lem:diff_circ_second}.
\setcounter{lem}{1}
\begin{lem}
    Consider the Schr\"{o}dinger equation $d \ket{u(t)}/dt = -i \mathcal{H} \ket{u(t)}$ such that the Hamiltonian $\mathcal{H}$ is given by Eq.~\eqref{eq:hamiltonian_diffop}.
    The time evolution operator $\exp(-i \mathcal{H} \tau)$ with the time increment $\tau$ can be approximated by the unitary $V^{(2)}$ in Eq.~\eqref{eq:trotter_second}. 
    The approximation error in the sense of the operator norm is upper bounded as 
    \begin{align}
        \left\Vert \exp(-i \mathcal{H} \tau) - V^{(2)} ( \gamma \tau, \lambda) \right\Vert &\leq \frac{\gamma^3 \tau^3}{6} (2n-3).
    \end{align}
\end{lem}

As mentioned in the main text, our task is to evaluate the commutators of the Hamiltonian.
Before calculating the commutators, we first derive some useful formula to calculate the commutators.
Here, we explicitly show the number of qubits $n$ to describe $s_j^{-}$ and $s_j^{+}$, notating $s_{j, n}^{-} := I^{\otimes n - j} \otimes \sigma_{01} \otimes \sigma_{10}^{\otimes j - 1}$ and $s_{j, n}^{+} := I^{\otimes n - j} \otimes \sigma_{10} \otimes \sigma_{01}^{\otimes j - 1}$, which will derive some useful formulas for proofs.
To begin with, we provide the following two lemmas which can be easily verified by the definition of $s^{-}_{j, n}$ and $s^{+}_{j, n}$.
\setcounter{lem}{6}
\begin{lem} \label{lem:s_recursive}
    Let $s_{j, n}^{-} := I^{\otimes n - j} \otimes \sigma_{01} \otimes \sigma_{10}^{\otimes j - 1}$ and $s_{j, n}^{+} := I^{\otimes n - j} \otimes \sigma_{10} \otimes \sigma_{01}^{\otimes j - 1}$.
    For $j \geq 2$ and $n \geq 3$, the following recursive relations hold:
    \begin{align}
        s_{j, n}^{-} = s_{j-1, n-1}^{-} \otimes \sigma_{10} \\
        s_{j, n}^{+} = s_{j-1, n-1}^{-} \otimes \sigma_{01}
    \end{align}
\end{lem}

\begin{lem} \label{lem:s_formula_1}
    Let $s_{j, n}^{-} := I^{\otimes n - j} \otimes \sigma_{01} \otimes \sigma_{10}^{\otimes j - 1}$ and $s_{j, n}^{+} := I^{\otimes n - j} \otimes \sigma_{10} \otimes \sigma_{01}^{\otimes j - 1}$.
    On the multiplication of $s_{j, n}^{-}$ and $s_{j, n}^{+}$, the following relationships hold: for $j > j' \geq 2$,
    \begin{align}
        \begin{cases}
            s^{-}_{j, n} s^{-}_{j', n} = s^{-}_{j', n} s^{-}_{j, n} = 0 \\
            s^{+}_{j, n} s^{+}_{j', n} = s^{+}_{j', n} s^{+}_{j, n} = 0 \\
            s^{-}_{j, n} s^{+}_{j', n} = s^{-}_{j', n} s^{+}_{j, n} = 0 \\
            s^{+}_{j, n} s^{-}_{j', n} = s^{+}_{j', n} s^{-}_{j, n} = 0,
        \end{cases}
    \end{align}
    for $j \geq 2$,
    \begin{align}
        \begin{cases}
            s^{-}_{j, n} s^{-}_{1, n} = s_{j-1, n-1}^{-} \otimes \sigma_{11} \\
            s^{-}_{1, n} s^{-}_{j, n} = s_{j-1, n-1}^{-} \otimes \sigma_{00} \\
            s^{+}_{j, n} s^{+}_{1, n} = s_{j-1, n-1}^{+} \otimes \sigma_{00} \\
            s^{+}_{1, n} s^{+}_{j, n} = s_{j-1, n-1}^{+} \otimes \sigma_{11} \\
            s^{-}_{j, n} s^{+}_{1, n} = s^{-}_{1, n} s^{+}_{j, n} = 0 \\
            s^{+}_{j, n} s^{-}_{1, n} = s^{+}_{1, n} s^{-}_{j, n} = 0,
        \end{cases}
    \end{align}
    and for $j \geq 1$,
    \begin{align}
        \begin{cases}
            s^{-}_{j, n} s^{-}_{j, n} = s^{+}_{j, n} s^{+}_{j, n} = 0 \\
            s^{-}_{j, n} s^{+}_{j, n} = I^{\otimes n - j} \otimes \sigma_{00} \otimes \sigma_{11}^{\otimes j - 1} \\
            s^{+}_{j, n} s^{-}_{j, n} = I^{\otimes n - j} \otimes \sigma_{11} \otimes \sigma_{00}^{\otimes j - 1}.
        \end{cases}
    \end{align}
\end{lem}

From lemmas~\ref{lem:s_recursive} and \ref{lem:s_formula_1}, we further obtain the following lemmas.
\begin{lem} \label{lem:s_formula_2}
    Let $s_{j, n} := e^{i\lambda} s^{-}_{j, n} + e^{-i\lambda} s^{+}_{j, n}$.
    On the multiplication of $s_{j, n}$, the following relationships hold:
    for $j > j' \geq 2$,
    \begin{align}
        s_{j, n} s_{j', n} = 0,
    \end{align}
    for $j \geq 2$,
    \begin{align}
        \begin{cases}
            s_{j, n} s_{1, n} = e^{2i\lambda} s_{j-1, n-1}^{-} \otimes \sigma_{11} + e^{-2i\lambda} s_{j-1, n-1}^{+} \otimes \sigma_{00} \\
            s_{1, n} s_{j, n} = e^{2i\lambda} s_{j-1, n-1}^{-} \otimes \sigma_{00} + e^{-2i\lambda} s_{j-1, n-1}^{+} \otimes \sigma_{11},
        \end{cases}
    \end{align}
    and for $j \geq 1$, 
    \begin{align}
        s_{j, n}^2 =& I^{\otimes n - j} \otimes \sigma_{00} \otimes \sigma_{11}^{\otimes j - 1} + I^{\otimes n - j} \otimes \sigma_{11} \otimes \sigma_{00}^{\otimes j - 1}.
    \end{align}
    In particular, $s_1^2 = I^{\otimes n}$, which implies that $s_1$ is unitary.
\end{lem}

Now, we prove Lemma~\ref{lem:diff_circ_second}
From Lemma~\ref{lem:s_formula_2}, $s_j$ and $s_{j'}$ commute for $j > j' \geq 2$.
Thus, let us group the terms of the Hamiltonian $\mathcal{H}$ in Eq.~\eqref{eq:hamiltonian_diffop} into $\mathcal{H}_1 := \gamma s_{1, n} $ and $H_2 := \gamma \sum_{j=2}^n s_{j, n}$ such that $\mathcal{H} = \mathcal{H}_1 + \mathcal{H}_2$.
The error of the second-order Suzuki formula is upper bounded~\cite[Proposition 9]{childs2021theory} by
\begin{align}
    &\left\Vert \exp(-i \mathcal{H} \tau) - V^{(2)}( \gamma \tau, \lambda ) \right\Vert \leq \frac{\tau^3}{12} \left\Vert \left[ \mathcal{H}_2, \left[ \mathcal{H}_2, \mathcal{H}_1 \right] \right] \right\Vert + \frac{\tau^3}{24} \left\Vert \left[ \mathcal{H}_1, \left[ \mathcal{H}_1, \mathcal{H}_2 \right] \right] \right\Vert. \label{eq:second_order_trotter}
\end{align}
To evaluate the upper bound, we rearrange the commutator of the first term in Eq.~\eqref{eq:second_order_trotter} using Lemmas~\ref{lem:s_recursive}, \ref{lem:s_formula_1}, and \ref{lem:s_formula_2}, as follows:
\begin{align}
    \left[ \mathcal{H}_2, \left[ \mathcal{H}_2, \mathcal{H}_1 \right] \right]
    &= \gamma^3 \sum_{j=2}^n \sum_{j'=2}^n \left[ s_{j, n}, \left[ s_{j', n}, s_{1, n} \right] \right] \nonumber \\
    &= \gamma^3 \sum_{j=2}^n \sum_{j'=2}^n \left[ e^{i\lambda} s^{-}_{j-1, n-1} \otimes \sigma_{10} + e^{-i\lambda} s^{+}_{j-1, n-1} \otimes \sigma_{01},  -e^{2i\lambda} s_{j'-1, n-1}^{-} \otimes Z + e^{-2i\lambda} s_{j'-1, n-1}^{+} \otimes Z \right] \nonumber \\
    &= \gamma^3 \sum_{j=3}^n \left( s^{-}_{j-1, n-1} s^{+}_{j-1, n-1} + s^{+}_{j-1, n-1} s^{-}_{j-1, n-1} \right) \otimes \left( e^{-i\lambda} \sigma_{10} + e^{i\lambda} \sigma_{01} \right) \nonumber \\
    &\quad - \gamma^3 \sum_{j=3}^n \left( e^{3i\lambda} 
    s_{j-2, n-2}^{-} \otimes I \otimes \sigma_{10} + e^{-3i\lambda} 
    s_{j-2, n-2}^{+} \otimes I \otimes \sigma_{01} \right) \nonumber \\
    &\quad - \gamma^3 \sum_{j'=3}^n \left( e^{3i\lambda} s^{-}_{j'-2, n-2} \otimes I \otimes \sigma_{10} + e^{-3i\lambda} s^{+}_{j'-2, n-2} \otimes I \otimes \sigma_{01} \right) \nonumber \\
    &\quad + \gamma^3 \left( e^{-i\lambda} I^{\otimes n - 1} \otimes \sigma_{10} + e^{i\lambda} I^{\otimes n - 1} \otimes \sigma_{01} \right) \nonumber \\
    &= \gamma^3 \sum_{j=3}^n I^{\otimes n - j} \otimes \left( \sigma_{00} \otimes \sigma_{11}^{\otimes j - 2} + \sigma_{11} \otimes \sigma_{00}^{\otimes j - 2} \right) \otimes \left( e^{-i\lambda} \sigma_{10} + e^{i\lambda} \sigma_{01} \right) \nonumber \\
    &\quad - 2\gamma^3 \sum_{j=3}^n \left( e^{3i\lambda} 
    s_{j-2, n-2}^{-} \otimes I \otimes \sigma_{10} + e^{-3i\lambda} 
    s_{j-2, n-2}^{+} \otimes I \otimes \sigma_{01} \right) \nonumber \\
    &\quad + \gamma^3 \left( e^{-i\lambda} I^{\otimes n - 1} \otimes \sigma_{10} + e^{i\lambda} I^{\otimes n - 1} \otimes \sigma_{01} \right)
\end{align}
Based on the similar calculation to Eq.~\eqref{eq:commutator_1}, we deduce that 
\begin{align}
    \left\Vert \left[ \mathcal{H}_2, \left[ \mathcal{H}_2, \mathcal{H}_1 \right] \right] \right\Vert &\leq \gamma^3 (3n-5).
\end{align}
Using Lemmas~\ref{lem:s_recursive}, \ref{lem:s_formula_1}, and \ref{lem:s_formula_2}, we also rearrange the commutator of the second term in Eq.~\eqref{eq:second_order_trotter}, as follows:
\begin{align}
    \left[ \mathcal{H}_1, \left[ \mathcal{H}_1, \mathcal{H}_2 \right] \right] &= \gamma^3 \sum_{j=2}^n \left[ s_{1, n}, \left[ s_{1, n}, s_{j, n} \right] \right] \nonumber \\
    &= \gamma^3 \sum_{j=2}^n \left( (s_{1, n}^2 s_{j, n} - s_{1, n} s_{j, n} s_{1, n}) - (s_{1, n} s_{j, n} s_{1, n} - s_{j, n} s_{1, n}^2) \right) \nonumber \\
    &= 2\gamma^3 \sum_{j=2}^n (s_{j, n} s_{1, n}^2 - s_{1, n} s_{j, n} s_{1, n}) \nonumber \\
    &= 2\gamma^3 \sum_{j=2}^n \left[ s_{j, n}, s_{1, n} \right] s_{1, n} \nonumber \\
    &= 2\gamma^3 \sum_{j=2}^n \left( -e^{2i\lambda} s_{j-1, n-1}^{-} \otimes Z + e^{-2i\lambda} s_{j-1, n-1}^{+} \otimes Z \right) s_{1, n} \nonumber \\
\end{align}
Based on the similar calculation to Eq.~\eqref{eq:commutator_1} and the fact that $s_{1, n}$ is unitary, we deduce that 
\begin{align}
    \left\Vert \left[ \mathcal{H}_1, \left[ \mathcal{H}_1, \mathcal{H}_2 \right] \right] \right\Vert &\leq 2 \gamma^3 (n-1).
\end{align}
Finally, we obtain 
\begin{align}
    \left\Vert \exp(-i \mathcal{H} \tau) - V^{(2)}( \gamma \tau, \lambda ) \right\Vert &\leq \frac{\tau^3}{12} \left\Vert \left[ \mathcal{H}_2, \left[ \mathcal{H}_2, \mathcal{H}_1 \right] \right] \right\Vert + \frac{\tau^3}{24} \left\Vert \left[ \mathcal{H}_1, \left[ \mathcal{H}_1, \mathcal{H}_2 \right] \right] \right\Vert. \nonumber \\
    &\leq \frac{\gamma^3 \tau^3}{6} (2n-3).
\end{align}

\section{Quantum circuit implementation and Trotter error} \label{sec:impl}

\subsection{Implementation and Trotter error for the advection equation} \label{sec:impl_advect}
Here, we consider the Hamiltonian $\mathcal{H}$ for the advection equation in one dimension as shown in Eq.~\eqref{eq:hamiltonian_advect} with the periodic BC, i.e.,
\begin{align}
    \mathcal{H} = -iv_1 D^{\pm}_\mathrm{P} = -\frac{iv_1}{2l} \left( \sum_{j=1}^n \left( s^{-}_j - s^{+}_j \right) + \sigma_{10}^{\otimes n} - \sigma_{01}^{\otimes n} \right).
\end{align}
Because the terms $-iv_1 \sum_{j=1}^n ( s^{-}_j - s^{+}_j ) /2l $ fall into Lemma~\ref{lem:diff_circ} by setting $\gamma=v_1/2l$ and $\lambda=-\pi/2$, we consider the terms for the periodic BC, which is rearranged as
\begin{align}
    \sigma_{10}^{\otimes n} - \sigma_{01}^{\otimes n} &= i \left( \frac{ \ket{0}^{\otimes n} - i \ket{1}^{\otimes n}}{\sqrt{2}} \frac{ \bra{0}^{\otimes n} + i \bra{1}^{\otimes n}}{\sqrt{2}} - \frac{ \ket{0}^{\otimes n} + i \ket{1}^{\otimes n}}{\sqrt{2}} \frac{ \bra{0}^{\otimes n} - i \bra{1}^{\otimes n}}{\sqrt{2}} \right) \nonumber \\
    &= i U_n\left( -\frac{\pi}{2} \right) \left(I \otimes X^{\otimes (n-1)} \right) \left( Z \otimes \ket{1} \! \bra{1}^{\otimes (n-1)} \right) \left(I \otimes X^{\otimes (n-1)} \right) U_n\left( -\frac{\pi}{2} \right)^\dagger.
\end{align}
Thus, the time evolution driven by this term reads
\begin{align}
    \exp \left( -i \left( -\frac{iv_1}{2l} (\sigma_{10}^{\otimes n} - \sigma_{01}^{\otimes n}) \right) \tau \right) &= U_n\left( -\frac{\pi}{2} \right) \left(I \otimes X^{\otimes (n-1)} \right) \crz^{1, \dots, n-1}_n \left( \frac{v_1 \tau}{l} \right) \left(I \otimes X^{\otimes (n-1)} \right) U_n\left( -\frac{\pi}{2} \right)^\dagger \nonumber \\
    &=: V_n \left( \frac{v_1 \tau}{2l}, -\frac{\pi}{2} \right).
\end{align}
Applying the first-order Lie-Trotter-Suzuki decomposition, we therefore obtain the time evolution operator $\exp(-i \mathcal{H} \tau)$ for the advection equation with the periodic boundary condition, as follows:
\begin{align}
    \exp \left(-i \mathcal{H} \tau \right) &\approx V \left( \frac{v_1 \tau}{2l}, -\frac{\pi}{2} \right) V_n \left( \frac{v_1 \tau}{2l}, -\frac{\pi}{2} \right).
\end{align}
The terms for the periodic BC commute with the other terms for $j>1$ in the Hamiltonian, as follows:
\begin{align}
    \left[ e^{i\lambda} s^{-}_j + e^{-i\lambda} s^{+}_j, \sigma_{10}^{\otimes n} - \sigma_{01}^{\otimes n} \right] &= \left[ e^{i\lambda} I^{\otimes (n-j)} \otimes \sigma_{01} \otimes \sigma_{10}^{\otimes (j-1)} + e^{-i\lambda} I^{\otimes (n-j)} \otimes \sigma_{10} \otimes \sigma_{01}^{\otimes (j-1)}, \right. \nonumber \\
    & \left. \quad \sigma_{10}^{\otimes (n-j)} \otimes \sigma_{10} \otimes \sigma_{10}^{\otimes (j-1)}  - \sigma_{01}^{\otimes (n-j)} \otimes \sigma_{01} \otimes \sigma_{01}^{\otimes (j-1)} \right] \nonumber \\
    &= 0.
\end{align}
For $j=1$, the commutator reads
\begin{align}
    \left[ e^{i\lambda} s^{-}_1 + e^{-i\lambda} s^{+}_1, \sigma_{10}^{\otimes n} - \sigma_{01}^{\otimes n} \right] &= \left[ e^{i\lambda} I^{\otimes (n-1)} \otimes \sigma_{01} + e^{-i\lambda} I^{\otimes (n-1)} \otimes \sigma_{10},  \sigma_{10}^{\otimes (n-1)} \otimes \sigma_{10} - \sigma_{01}^{\otimes (n-1)} \otimes \sigma_{01} \right] \nonumber \\
    &= \left( e^{i\lambda} \sigma_{10}^{\otimes (n-1)} + e^{-i\lambda} \sigma_{01}^{\otimes (n-1)} \right) \otimes Z,
\end{align}
and its operator norm is one, which is verified by similar calculation to that of the proof of Lemma~\ref{lem:diff_circ}.
Therefore, the approximation error of the Trotter decomposition is upper bounded as
\begin{align}
    &\left\Vert \exp(-i \mathcal{H} \tau) - V \left( \frac{v_1 \tau}{2l}, -\frac{\pi}{2} \right) V_n \left( \frac{v_1 \tau}{2l}, -\frac{\pi}{2} \right) \right\Vert \nonumber \\ 
    &\leq \frac{v_1^2 \tau^2}{8 l^2} \left\Vert \left[ \sum_{j=2}^n \left( e^{-i \pi / 2} s^{-}_j + e^{i \pi / 2} s^{+}_j \right) -i (\sigma_{10}^{\otimes n} - \sigma_{01}^{\otimes n}), e^{-i \pi / 2} s^{-}_1 + e^{i \pi / 2} s^{+}_1 \right] \right\Vert \nonumber \\
    &\leq \frac{v_1^2 \tau^2}{8 l^2} \left( \left\Vert \left[ \sum_{j=2}^n \left( e^{-i \pi / 2} s^{-}_j + e^{i \pi / 2} s^{+}_j \right), e^{-i \pi / 2} s^{-}_1 + e^{i \pi / 2} s^{+}_1 \right] \right\Vert + \left\Vert \left[ -i (\sigma_{10}^{\otimes n} - \sigma_{01}^{\otimes n}), e^{-i \pi / 2} s^{-}_1 + e^{i \pi / 2} s^{+}_1 \right] \right\Vert \right) \nonumber \\
    &\leq \frac{v_1^2 \tau^2}{8 l^2} \left( \sum_{j=2}^n \left\Vert \left[ e^{-i \pi / 2} s^{-}_j + e^{i \pi / 2} s^{+}_j, e^{-i \pi / 2} s^{-}_1 + e^{i \pi / 2} s^{+}_1 \right] \right\Vert + \left\Vert \left[ -i (\sigma_{10}^{\otimes n} - \sigma_{01}^{\otimes n}), e^{-i \pi / 2} s^{-}_1 + e^{i \pi / 2} s^{+}_1 \right] \right\Vert \right) \nonumber \\
    &= \frac{v_1^2 \tau^2 n}{8 l^2}.
\end{align}
We can also obtain the second-order formula as
\begin{align}
    V^{(2)}\left( \frac{v_1 \tau}{2l}, -\frac{\pi}{2} \right) := V_n \left( \frac{v_1 \tau}{4l}, -\frac{\pi}{2} \right) V \left( \frac{v_1 \tau}{2l}, -\frac{\pi}{2} \right) V_n \left( \frac{v_1 \tau}{4l}, -\frac{\pi}{2} \right),
\end{align}
whose approximation error is upper bounded by
\begin{align}
    \left\Vert \exp(-i \mathcal{H} \tau) - V^{(2)}\left( \frac{v_1 \tau}{2l}, -\frac{\pi}{2} \right) \right\Vert 
    &\leq \frac{v_1^3 \tau^3 (2n-1)}{48 l^3},
\end{align}
based on the the similar calculation to that in the proof of Lemma~\ref{lem:diff_circ_second}.
For a $d$-dimensional case, we can easily extend the discussion here based on the Lemma~\ref{lem:diff_circ_d-dim}.

\subsection{Implementation and Trotter error for the wave equation} \label{sec:impl_wave}
Here, we consider the following Hamiltonian $\mathcal{H}$ for the wave equation in one dimension.
The Hamiltonian $\mathcal{H}$ is discretized by the forward and backward difference operators so that the resulting $\mathcal{H}$ can be actually Hermitian matrix, as follows:
\begin{align}
    \mathcal{H} &= c \left( \sigma_{01} \otimes D^{+}_\mathrm{D} - \sigma_{10} \otimes D^{-}_\mathrm{D} \right) \nonumber \\
    &= \frac{c}{l} \left( \sum_{j=1}^n (\sigma_{01} \otimes s^{-}_j + \sigma_{10} \otimes s^{+}_j) - (\sigma_{01} + \sigma_{10}) \otimes I^{\otimes n} \right).
\end{align}
By similar calculation to Eq.~\eqref{eq:exp_iHt}, we obtain
\begin{align}
    \exp (-i \mathcal{H} \tau) &\approx \rx_{n+1} \left(-\frac{2c\tau}{l} \right) \prod_{j=1}^{n} \tilde{U}^{1,\dots, j}_{n+1}(0) X_j \crz^{1, \dots, j}_{n+1} \left( \frac{2c\tau}{l} \right) X_j \tilde{U}^{1,\dots, j}_{n+1}(0)^\dagger, \nonumber \\
    &=: \tilde{V}\left( \frac{c\tau}{l}, 0 \right)
\end{align}
where $\rx_{n+1}$ and $X_{n+1}$ are the RX and X gates acting on the $(n+1)$-th qubit, respectively, and 
\begin{align}
    \tilde{U}^J_{n+1}(\lambda) := \left( \prod_{m \in J} \cnot^{n+1}_m \right) P_{n+1}(\lambda) H_{n+1},
\end{align}
which is the extension of the unitary $U_j(\lambda)$ in Eq.~\eqref{eq:bell_basis}.
Furthermore, we obtain the upper bound of the Trotter error, as follows:
\begin{align}
    \left\Vert \exp(-i \mathcal{H} \tau) - \tilde{V}\left( \frac{c\tau}{l}, 0 \right) \right\Vert 
    &\leq \frac{c^2 \tau^2 n}{2l^2},
\end{align}
which can be obtained by the similar calculation to that in the proof of Lemma~\ref{lem:diff_circ}.
We can also obtain the second-order formula as
\begin{align}
    \tilde{V}^{(2)}\left( \frac{c\tau}{l}, 0 \right) := \rx_{n+1} \left(\frac{c\tau}{l} \right) \tilde{V}\left( \frac{c\tau}{l}, 0 \right) \rx_{n+1} \left(-\frac{c\tau}{l} \right),
\end{align}
whose error is upper bounded by
\begin{align}
    \left\Vert \exp(-i \mathcal{H} \tau) - \tilde{V}^{(2)}\left( \frac{c\tau}{l}, 0 \right) \right\Vert 
    &\leq \frac{c^3 \tau^3 (2n-1)}{6l^3},
\end{align}
based on the the similar calculation to that in the proof of Lemma~\ref{lem:diff_circ_second}.

\bibliographystyle{unsrt}
\bibliography{ref}

\end{document}